\documentclass[11pt,notitlepage]{article}%
\usepackage{latexsym,epsfig,amsmath}
\usepackage{epsfig,graphicx}
\usepackage{amsfonts}
\usepackage{tabularx}
\usepackage{rotating}
\usepackage{harvard}
\usepackage{hyperref}
\usepackage{amsmath}
\usepackage{amssymb}
\usepackage{graphicx}
\usepackage{color}%
\providecommand{\U}[1]{\protect\rule{.1in}{.1in}}
\newtheorem{theorem}{Theorem}

\newtheorem{corollary}[theorem]{Corollary}

\newtheorem{lemma}[theorem]{Lemma}

\newcounter{assumec}

\newtheorem{condition}{Condition}
\newenvironment{proof}[1][Proof]{\textbf{#1.} }{\ \rule{0.5em}{0.5em}}
\evensidemargin\oddsidemargin
\begin{document}

\title{Invariance Principles for Dependent Processes Indexed by Besov Classes with an
Application to a Hausman Test for Linearity}
\author{Guido M. Kuersteiner\thanks{University of Maryland, Department of Economics,
Tydings Hall 3145, 7343 Prinkert Dr., College Park, MD 20742, USA. email:
kuersteiner@econ.umd.edu; http://econweb.umd.edu/\symbol{126}kuersteiner/}%
\thanks{Very helpful comments received from the editor, Whitney Newey, and two
anonymous referees are gratefully acknowledged. }\\University of Maryland}
\date{}
\maketitle

\begin{abstract}
This paper considers functional central limit theorems for stationary
absolutely regular mixing processes. Bounds for the entropy with bracketing
are derived using recent results in Nickl and P\"{o}tscher (2007). More
specifically, their bracketing metric entropy bounds are extended to a norm
defined in Doukhan, Massart and Rio (1995, henceforth DMR) that depends both
on the marginal distribution of the process and on the mixing coefficients.
Using these bounds, and based on a result in DMR, it is shown that for the
class of weighted Besov spaces polynomially decaying tail behavior of the
function class is sufficient to obtain a functional central limit theorem
under minimal dependence conditions. A second class of functions that allow
for a functional central limit theorem under minimal conditions are smooth
functions defined on bounded sets. Similarly, a functional CLT for
polynomially explosive tail behavior is obtained under additional moment
conditions that are easy to check. An application to a Hausman (1978)
specification test for linearity of the conditional mean illustrates the theory.

\bigskip\bigskip


\end{abstract}

\newpage

\section{Introduction}

This paper studies central limit theorems for empirical processes defined on
dependent data and indexed by smooth classes of functions. Doukhan, Massart
and Rio (1994) and Doukhan, Massart and Rio (1995) (henceforth DMR) are
landmark contributions in this literature. The key insight from those papers
is that a specific norm that combines dependence properties and the marginal
distribution of the process provides the appropriate measure to assess the
complexity of the function class in terms of bracketing entropy. However, as
pointed out by Rio (1998, 2013) the results of DMR are not minimal in the
sense of providing convergence under dependence assumptions equivalent to
finite dimensional cases. In fact, for a $\beta$-mixing process with mixing
coefficients $\beta_{m},$ central limit theorems can be established under the
minimal condition that $\sum_{m=0}^{\infty}\beta_{m}<\infty.$ Rio (1998, 2013)
shows that such minimal results are possible in some cases involving
Vapnik-Cervonenkis (VC) classes as well as certain Lipschitz type functions.
In this paper the function classes for which such minimal results are possible
are expanded to smooth classes of rapidly asymptoting functions as well as
function classes defined on a bounded set. This is achieved by directly
employing recent results of complexity measures for weighted Besov spaces in
Haroske and Triebel (2005) and Nickl and P\"{o}tscher (2007). In addition to
these improvements over the existing literature the paper also gives a number
of explicit results that relate dependence properties of the underlying
process to smoothness properties of the indexing function class.

Separate results then need to be employed to arrive at explicit central limit
theorems. This is particularly relevant for dependent data where there is a
potentially complex interaction between the properties of the function class,
dependence of the process and properties of the marginal distribution of the
process. An additional requirement, especially in econometric applications, is
that function spaces are defined on unbounded sets, typically $\mathbb{R}%
^{d}.$ This further limits applicability of many results available in the iid literature.

Andrews\ (1991) has given similar results under related conditions for
processes that are not necessarily stationary. However, Andrews (1991)
essentially is limited to function classes defined on a bounded domain. This
paper compliments Andrews (1991) by allowing for weaker assumptions on the
support of function classes while assuming stationarity and slightly stronger
mixing conditions. Nickl (2007) mentions the possibility of obtaining explicit
empirical process central limit theorems for the dependent case using the
approach pursued here but does not give such results. A useful by-product of
obtaining empirical central limit theorems for specific function classes are
stochastic equicontinuity results for these function classes. This fact is
exploited in the part of the paper that develops a Hausman specification test
for the conditional mean function.

Empirical central limit theorems have a long history in probability and have
found wide applications in statistics. Early results are due to Dudley (1978,
1984) and Pollard (1982). General results for iid data using bracketing were
obtained by Ossiander (1987) and Pollard (1989). Results based on VC classes
are due to Pollard (1990). Early results for dependent processes include
Berkes and Phillip (1977) generalizing Donsker's theorem to strongly mixing
stationary sequences. Uniform CLT's over function classes for dependent
processes were studied in Doukhan, Leon and Portal (1987), Massart (1987),
Andrews (1991), Andrews and Pollard (1994) and Hansen (1996). Arcones and Yu
(1994) consider absolutely regular processes indexed by VC classes. A very
influential paper is Doukhan, Massart and Rio (1995) which considers
absolutely regular processes under a bracketing condition, extending results
by Ossiander to the dependent case.

The paper is organized as follows. Section \ref{Existing} presents definitions
for measures of dependence and discusses the existing functional CLT's
underlying the developments of the paper. Section \ref{Function Spaces}
defines the smooth classes of function spaces considered. Section
\ref{Sec_CLT} presents the main theory and contains a detailed comparison with
other related results in the literature. An application to the problem of
testing the specification of the conditional mean using a Hausman test is
presented in Section \ref{Sec_Hausman}. Proofs are collected in the appendix
in Section \ref{Appendix}.

\section{Notation and Existing Results\label{Existing}}

The sequence $\chi_{t}$ consists of (measurable) random variables defined on
the probability space $\left(  \Omega,\mathcal{A},\mathbb{P}\right)  .$ Assume
that $\left\{  \chi_{t}\right\}  _{t=-\infty}^{\infty}$ is strictly stationary
with values in the measurable space $\left(  \mathbb{R}^{d},\mathcal{B}%
^{d}\right)  $ where $\mathcal{B}^{d}$ is the Borel $\sigma$-field on
$\mathbb{R}^{d}$ and $d\in\mathbb{N}_{+}$. Let $\mathcal{A}^{l}=\sigma\left(
\chi_{t}:t\leq l\right)  $ be the sigma field generated by $,...\chi
_{l-1},\chi_{l}$ and similarly $\mathcal{D}^{l}=\sigma\left(  \chi_{t}:t\geq
l\right)  .$ Following DMR, p.379 the absolutely regular mixing coefficient
$\beta_{m}$ is defined as
\[
2\beta_{m}=\sup\sum_{\left(  i,j\right)  \in I\times J}\left\vert
\mathbb{P}\left(  A_{i}\cap D_{j}\right)  -\mathbb{P}\left(  A_{i}\right)
\mathbb{P}\left(  D_{j}\right)  \right\vert
\]
where the supremum is taken over all finite partitions $A_{i}$ and $D_{j}$ of
$\mathcal{A}^{0}$ and $\mathcal{D}^{m}.$ The definition of $\beta_{m}$ is due
to Volkonski and Rozanov (1959) who give an alternative equivalent formulation
that is sometimes used in the literature (see for example Arcones and Yu,
1994). Strong mixing is defined as
\[
\alpha_{m}=\sup_{\left(  A,D\right)  \in\mathcal{A}^{0}\times\mathcal{D}^{m}%
}\left\vert \mathbb{P}\left(  D\cap A\right)  -\mathbb{P}\left(  A\right)
\mathbb{P}\left(  D\right)  \right\vert ,
\]
and $\varphi$-mixing is based on
\[
\varphi_{m}=\sup_{\left(  A,D\right)  \in\mathcal{A}^{0}\times\mathcal{D}^{m}%
}\left\vert \mathbb{P}\left(  D|A\right)  -\mathbb{P}\left(  D\right)
\right\vert .
\]
The relationship $2\alpha_{m}\leq\beta_{m}\leq\varphi_{m}\leq1$ holds. The
condition
\begin{equation}%
{\textstyle\sum\nolimits_{m=0}^{\infty}}
\beta_{m}<\infty\label{Cond_beta_sum1}%
\end{equation}
is frequently imposed in what follows.

Define the Euclidian norm for a real valued matrix or vector $A$ as
$\left\Vert A\right\Vert ^{2}=\operatorname*{tr}AA^{\prime}.$ Let
$\mathcal{X}\subseteq\mathbb{R}^{d}$ be a non-empty Borel set. Define the
sup-norm $\left\Vert f\right\Vert _{\infty}=\sup_{x\in\mathcal{X}}\left\vert
f\left(  x\right)  \right\vert $ for any measurable function $f:\mathcal{X}%
\mathcal{\rightarrow}\mathbb{R}$. Similarly, for $r\geq1$ let $\left\Vert
f\right\Vert _{r,P}=\left(  \int\left\vert f\left(  x\right)  \right\vert
^{r}dP\left(  x\right)  \right)  ^{1/r}$ where $P$ is the marginal
distribution of $\chi_{t}$ and let $\mathcal{L}_{r}\left(  P\right)  $ be the
set of functions with $\left\Vert f\right\Vert _{r,P}<\infty.$ The following
definitions are given in Rio (1993) and DMR. For a nonincreasing function
$h:\mathbb{R\rightarrow R}$ define the inverse $h^{-1}\left(  u\right)
=\inf\left\{  t:h\left(  t\right)  \leq u\right\}  .$ Let $Q_{f}\left(
u\right)  $ be the quantile function defined as the inverse of the tail
probability $P\left(  \left\vert f\left(  \chi_{t}\right)  \right\vert
>t\right)  .$ Let $\left\lfloor t\right\rfloor $ be the largest integer
smaller or equal to $t\in\mathbb{R}$ and define $\beta^{-1}\left(  u\right)
=\inf\left\{  t:\beta_{\left\lfloor t\right\rfloor }\leq u\right\}  .$ Now
define the norm
\[
\left\Vert f\right\Vert _{2,\beta}=\sqrt{\int_{0}^{1}\beta^{-1}\left(
u\right)  \left(  Q_{f}\left(  u\right)  \right)  ^{2}du}<\infty.
\]
DMR, Lemma 1, show that if (\ref{Cond_beta_sum1}) holds, the set
$\mathcal{L}_{2,\beta}\left(  P\right)  $ of functions with $\left\Vert
f\right\Vert _{2,\beta}<\infty$ equipped with the norm $\left\Vert
.\right\Vert _{2,\beta}$ is a normed subspace of $\mathcal{L}_{2}\left(
P\right)  $ and that $\left\Vert f\right\Vert _{2,P}\leq\left\Vert
f\right\Vert _{2,\beta}.$ DMR (p.401) remark that $%
{\textstyle\sum\nolimits_{m=0}^{\infty}}
\beta_{m}=\int_{0}^{1}\beta^{-1}\left(  u\right)  du.$ This implies that under
the summability condition in (\ref{Cond_beta_sum1}) the space $\mathcal{L}%
_{2,\beta}\left(  P\right)  $ contains the space of bounded functions
$\mathcal{L}_{\infty}\left(  P\right)  .$ A reverse conclusion of their remark
is that for bounded functions, $f\in\mathcal{L}_{2,\beta}\left(  P\right)  $
implies that $%
{\textstyle\sum\nolimits_{m=0}^{\infty}}
\beta_{m}<\infty$ needs to hold.

Consider the class of functions $\mathcal{F}$ with elements $f:\mathcal{X}%
\rightarrow\mathbb{R}$. For a sample $\left\{  \chi_{t}\right\}  _{t=1}^{n}$
define the empirical process
\[
v_{n}\left(  f\right)  =n^{-1/2}\sum_{t=1}^{n}\left(  f\left(  \chi
_{t}\right)  -E\left[  f\left(  \chi_{t}\right)  \right]  \right)  .
\]
When (\ref{Cond_beta_sum1}) is satisfied, Rio (1993, Theorem 1.2) shows that
for $f\in\mathcal{L}_{2,\beta}\left(  P\right)  ,$%
\[
\sum_{t=-\infty}^{\infty}\left\vert \operatorname*{Cov}\left(  f\left(
\chi_{0}\right)  ,f\left(  \chi_{t}\right)  \right)  \right\vert
\leq4\left\Vert f\right\Vert _{2,\beta}^{2}%
\]
and for $\Gamma\left(  f,f\right)  =\sum_{t=-\infty}^{\infty}%
\operatorname*{Cov}\left(  f\left(  \chi_{0}\right)  ,f\left(  \chi
_{t}\right)  \right)  $ it follows that
\[
\lim_{n\rightarrow\infty}\operatorname*{Var}\left(  v_{n}\left(  f\right)
\right)  =\Gamma\left(  f,f\right)  \leq4\left\Vert f\right\Vert _{2,\beta
}^{2}.
\]

Following DMR and van der Vaart and Wellner (1996, p.83) let $\mathcal{F}$ be
a subset of a normed space $\left(  V,\left\Vert {}\right\Vert _{V}\right)  $
of functions $f:\mathcal{X}\rightarrow\mathbb{R}$ with norm $\left\Vert
{}\right\Vert _{V}.$ For any pair of functions, $l,u\in\mathcal{F}$ and
$\delta>0$, the set $\left[  l,u\right]  \subset\mathcal{F}$ is a $\delta
$-bracket if $l\leq u$ with $\left\Vert l-u\right\Vert _{V}\leq\delta$ and for
all $f\in\left[  l,u\right]  $ it follows that $l\leq f\leq u.$ The bracketing
number $N_{[]}\left(  \delta,\mathcal{F},\left\Vert {}\right\Vert _{V}\right)
$ is the smallest number of $\delta$-brackets needed to cover $\mathcal{F}$.
The entropy with bracketing is the logarithm of $N_{[]}\left(  \delta
,\mathcal{F},\left\Vert {}\right\Vert _{V}\right)  $ denoted by $H_{[]}\left(
\delta,\mathcal{F},\left\Vert {}\right\Vert _{V}\right)  .$

DMR establish the following Theorem, see DMR, Theorem 1:

\begin{theorem}
[Doukhan, Massart and Rio, 1995]\label{Theorem_DMR}Assume that $\chi_{t}$ is a
strictly stationary $\beta$-mixing sequence with (\ref{Cond_beta_sum1})
holding, marginal distribution $P$ and $\mathcal{F}$ a class of functions $f$
with $\mathcal{F\subset L}_{2,\beta}\left(  P\right)  $ such that
\begin{equation}
\int_{0}^{1}\sqrt{H_{[]}\left(  \delta,\mathcal{F},\left\Vert {}\right\Vert
_{2,\beta}\right)  }d\delta<+\infty. \label{Bracketing}%
\end{equation}
Then the series $\sum_{t\in\mathbf{Z}}\operatorname*{Cov}\left(  f\left(
\chi_{0}\right)  ,f\left(  \chi_{t}\right)  \right)  $ is absolutely
convergent over $\mathcal{F}$ to a nonnegative quadratic form $\Gamma\left(
f,f\right)  $ and $\left(  \Gamma\left(  f,f\right)  \right)  ^{1/2}%
=\left\Vert f\right\Vert _{\Gamma}\leq2\left\Vert f\right\Vert _{2,\beta}.$ In
addition there exists a sequence $\left(  v^{\left(  n\right)  }\right)
_{n>0}$ of Gaussian processes indexed by $\mathcal{F}$ with covariance
function $\Gamma$ and a.s. uniformly continuous sample paths such that
\[
\sup_{f\in\mathcal{F}}\left\vert v_{n}\left(  f\right)  -v^{\left(  n\right)
}\left(  f\right)  \right\vert \rightarrow_{p}0\text{ as }n\rightarrow\infty.
\]

\end{theorem}

The proof of Theorem \ref{Theorem_DMR} is given on p.409 of DMR and involves
showing the convergence of the finite dimensional distributions as well as
establishing a stochastic equicontinuity property. The finite dimensional
vector $v_{n}\left(  f_{1}\right)  ,...,v_{n}\left(  f_{k}\right)  $ converges
weakly by a result of Doukhan, Massart and Rio (2014) such that%
\begin{equation}
\left(  v_{n}\left(  f_{1}\right)  ,...,v_{n}\left(  f_{k}\right)  \right)
\rightarrow^{d}\left(  v\left(  f_{1}\right)  ,...,v\left(  f_{k}\right)
\right)  , \label{Fidi}%
\end{equation}
where $v\left(  f\right)  $ is a Gaussian process with covariance function
$\Gamma$ and a.s. uniformly continuous sample paths. The asymptotic
equicontinuity condition is established by DMR (see p.410) and states that for
every $\epsilon>0$:
\begin{equation}
\lim_{\delta\rightarrow0}\underset{n\rightarrow\infty}{\lim\sup}%
\mathbb{P}^{\ast}\left(  \sup_{\left\Vert f-g\right\Vert _{2,\beta}\leq
\delta,\text{ }f,g\in\mathcal{F}}\left\vert v_{n}\left(  f\right)
-v_{n}\left(  g\right)  \right\vert >\epsilon\right)  =0,
\label{StochEquicont}%
\end{equation}
where $\mathbb{P}^{\ast}$ is outer probability. The short hand notation
$v_{n}\left(  f\right)  \rightsquigarrow v\left(  f\right)  $ is used when
both (\ref{Fidi}) and (\ref{StochEquicont}) hold. The implication of Theorem
\ref{Theorem_DMR} that (\ref{StochEquicont}) holds as a sufficient condition
for the statement of theorem is of independent interest in this paper and will
be used in Section \ref{Sec_Hausman} for the analysis of semiparametric
econometric procedures.

Theorem \ref{Theorem_DMR} delivers a functional central limit theorem under
close to minimal conditions on dependence and high level assumptions regarding
the function classes it covers. Doukhan, Massart and Rio (1994) give a counter
example where finite dimensional weak convergence fails for a function not in
$\mathcal{L}_{2,\beta}\left(  P\right)  .$ Rio (2013, p.104) notes that DMR's
CLT holds for $\beta_{m}=O\left(  m^{-b}\right)  $ with $b>1$ and thus does
not quite achieve minimal conditions on dependence. Rio (1998, 2013) provides
functional CLT's for VC classes of functions as well as for classes of
functions satisfying certain bracketing conditions under the minimal
dependence assumption in (\ref{Cond_beta_sum1}).

The difficulties of obtaining results under minimal dependence assumptions as
well as in applying the central limit theorem to particular statistical
problems are related to verifying (\ref{Bracketing}) for specific function
classes. The bracketing integral convolutes conditions related to the
dependence of the process, the marginal distribution of the process, tail
behavior of the function class and smoothness restrictions of the function
class into a single integrability condition. This paper extends results by
Nickl and P\"{o}tscher (2007) on the bracketing properties of function spaces
to disentangle these restrictions into conditions that can be individually
verified in an application. In some cases this approach leads to functional
CLT's under minimal dependence conditions.

In work preceding DMR, Ossiander (1987) obtains a version of Theorem
\ref{Theorem_DMR} under independence. In that case, the bracketing integral is
with respect to the $L_{2,P}$ norm $\left\Vert .\right\Vert _{2,P}$. In
applications one still needs to determine function classes that satisfy
(\ref{Bracketing}). Specific results for this case were obtained by Nickl and
P\"{o}tscher (2007) who also provide references to the previous literature.

\section{Function Spaces\label{Function Spaces}}

The purpose of this section is to introduce the function spaces for which the
bracketing condition in Theorem \ref{Theorem_DMR} is verified. The most
general class of function spaces considered are Besov spaces. Of particular
importance are weighted Besov spaces which provide a mechanism to handle
functions with unbounded support that do not vanish in the tails. Special
cases of Besov spaces such as Sobolev, H\"{o}lder and Zygmund spaces are
introduced subsequently.

The definition of Besov spaces follows Nickl and P\"{o}tscher (2007, Remark
2). For Lebesgue measure $\lambda$ let $\mathcal{L}_{p}\left(  \mathbb{R}%
^{d},\lambda\right)  $ be the set of all functions $f:\mathbb{R}%
^{d}\rightarrow\mathbb{R}$ with $\left\Vert f\right\Vert _{p,\lambda}=\left(
\int\left\vert f\left(  x\right)  \right\vert ^{p}dx\right)  ^{1/p}<\infty.$
Let $\alpha=\left(  \alpha_{1},...,\alpha_{d}\right)  $ be a multi index of
non-negative integers $\alpha_{i,}$ with $\left\vert \alpha\right\vert
=\sum_{i=1}^{d}\alpha_{i}$ and let $D^{\alpha}$ denote the partial
differential operator $\partial^{\left\vert \alpha\right\vert }/\left(
\left(  \partial x_{1}\right)  ^{\alpha_{1}}...\left(  \partial x_{d}\right)
^{\alpha_{d}}\right)  $ of order $\left\vert \alpha\right\vert $ in the sense
of distributions - see Stein (1970, p. 121). For a function $f:\mathbb{R}%
^{d}\rightarrow\mathbb{R}$ the difference operator $\Delta_{z}$ is defined as
$\Delta_{z}f\left(  .\right)  =f\left(  .+z\right)  -f\left(  .\right)  $ and
$\Delta_{z}^{2}f\left(  .\right)  =\Delta_{z}\left(  \Delta_{z}f\left(
.\right)  \right)  $ for $z\in\mathbb{R}^{d}.$ Let $0<s<\infty$ and set
$s=\left[  s\right]  ^{-}+\left\{  s\right\}  ^{+}$ where $\left[  s\right]
^{-}$ is integer and $0<\left\{  s\right\}  ^{+}\leq1.$ For example, when
$s=1,$ $\left\{  s\right\}  ^{+}=1$ and $\left[  s\right]  ^{-}=0.$ Let $1\leq
p\leq\infty$ and $1\leq q\leq\infty.$ For $f\in$ $\mathcal{L}_{p}\left(
\mathbb{R}^{d},\lambda\right)  $ with $\left\Vert D^{\alpha}f\right\Vert
_{p,\lambda}<\infty$ and for $0\leq\alpha\leq\left[  s\right]  ^{-}$ define
\[
\left\Vert f\right\Vert _{s,p,q,\lambda}^{\ast}=\sum_{0\leq\alpha\leq\left[
s\right]  ^{-}}\left\Vert D^{\alpha}f\right\Vert _{p,\lambda}+\sum
_{\alpha=\left[  s\right]  ^{-}}\left(  \int_{\mathbb{R}^{d}}\left\vert
z\right\vert ^{-\left\{  s\right\}  ^{+}q-d}\left\Vert \Delta_{z}^{2}%
D^{\alpha}f\right\Vert _{p,\lambda}^{q}dz\right)  ^{1/q}%
\]
for $q<\infty,$ and for $q=\infty$ define%
\[
\left\Vert f\right\Vert _{s,p,\infty,\lambda}^{\ast}=\sum_{0\leq\alpha
\leq\left[  s\right]  ^{-}}\left\Vert D^{\alpha}f\right\Vert _{p,\lambda}%
+\sum_{\alpha=\left[  s\right]  ^{-}}\sup_{0\neq z\in\mathbb{R}^{d}}\left\vert
z\right\vert ^{-\left\{  s\right\}  ^{+}}\left\Vert \Delta_{z}^{2}D^{\alpha
}f\right\Vert _{p,\lambda}.
\]
The Besov space $\mathcal{B}_{pq}^{s}\left(  \mathbb{R}^{d}\right)  $ is
defined as $\mathcal{B}_{pq}^{s}\left(  \mathbb{R}^{d}\right)  =\left\{
f\in\mathcal{L}_{p}\left(  \mathbb{R}^{d},\lambda\right)  :\left\Vert
f\right\Vert _{s,p,q,\lambda}^{\ast}<\infty\right\}  .$ An equivalent
definition can be given in terms of Fourier transforms $F$ acting on the space
of complex tempered distributions on $\mathbb{R}^{d}$ (see Edmunds and
Triebel, 1996, 2.2.1). Denote by $F^{-1}$ the inverse of $F.$ Let $\varphi
_{0}\left(  x\right)  $ be a complex valued $C^{\infty}$-function on
$\mathbb{R}^{d}$ with $\varphi_{0}\left(  x\right)  =1$ if $\left\Vert
x\right\Vert \leq1$ and $\varphi_{0}\left(  x\right)  =0$ if $\left\Vert
x\right\Vert \geq3/2.$ Define $\varphi_{1}\left(  x\right)  =\varphi
_{0}\left(  x/2\right)  -\varphi_{0}\left(  x\right)  $ and $\varphi
_{k}\left(  x\right)  =\varphi_{1}\left(  2^{-k+1}x\right)  $ for
$k\in\mathbb{N}$. Let $0\leq s<\infty,$ $1\leq p\leq\infty,$ $1\leq
q\leq\infty,$ with $q=1$ if $s=0.$ For $f\in$ $\mathcal{L}_{p}\left(
\mathbb{R}^{d},\lambda\right)  $ and $q<\infty$ define
\[
\left\Vert f\right\Vert _{s,p,q,\lambda}=\left(  \sum_{k=0}^{\infty}%
2^{ksq}\left\Vert F^{-1}\left(  \varphi_{k}Ff\right)  \right\Vert _{p,\lambda
}^{q}\right)  ^{1/q}%
\]
and for $q=\infty$%
\[
\left\Vert f\right\Vert _{s,p,\infty,\lambda}=\sup_{0\leq k<\infty}%
2^{ks}\left\Vert F^{-1}\left(  \varphi_{k}Ff\right)  \right\Vert _{p,\lambda
}.
\]
Then, it follows (see Nickl and P\"{o}tscher, 2007, p. 180) that
\[
\mathcal{B}_{pq}^{s}\left(  \mathbb{R}^{d}\right)  =\left\{  f\in
\mathcal{L}_{p}\left(  \mathbb{R}^{d},\lambda\right)  :\left\Vert f\right\Vert
_{s,p,q,\lambda}<\infty\right\}
\]
and the norms $\left\Vert f\right\Vert _{s,p,q,\lambda}^{\ast}$ and
$\left\Vert f\right\Vert _{s,p,q,\lambda}$ are equivalent on $\mathcal{B}%
_{pq}^{s}\left(  \mathbb{R}^{d}\right)  .$ Define $\left\langle x\right\rangle
=1+\left\Vert x\right\Vert ^{2}.$ Weighted Besov spaces are now defined as in
Edmunds and Triebel (1996, 4.2) and Nickl and P\"{o}tscher (2007, p.181) for
$\vartheta\in\mathbb{R}$ as
\[
\mathcal{B}_{pq}^{s}\left(  \mathbb{R}^{d},\vartheta\right)  =\left\{
f:\left\Vert f\left(  .\right)  \left\langle x\right\rangle ^{\vartheta
/2}\right\Vert _{s,p,q,\lambda}<\infty\right\}  .
\]
For $s>d/p$ or $s=d/p$ with $q=1$ define%
\[
B_{pq}^{s}\left(  \mathbb{R}^{d},\vartheta\right)  =\mathcal{B}_{pq}%
^{s}\left(  \mathbb{R}^{d},\vartheta\right)  \cap\left\{  f:f\left(  .\right)
\left\langle x\right\rangle ^{\vartheta/2}\in C\left(  \mathbb{R}^{d}\right)
\right\}
\]
where $C\left(  \mathbb{R}^{d}\right)  $ is the vector space of bounded
continuous real valued functions on $\mathbb{R}^{d}$ with the sup-norm
$\left\Vert .\right\Vert _{\infty}.$ Nickl and P\"{o}tscher (2007, Proposition
3) show that $f\in B_{pq}^{s}\left(  \mathbb{R}^{d}\right)  $ implies that $f$
is bounded and if $p<\infty$ it also follows that $\lim_{\left\Vert
x\right\Vert \rightarrow\infty}f\left(  x\right)  =0.$ These restrictions do
not necessarily apply when $f\in B_{pq}^{s}\left(  \mathbb{R}^{d}%
,\vartheta\right)  $ and $\vartheta<0$. This feature of weighted spaces is
important for applications in econometrics, as will be demonstrated in Section
\ref{Sec_Hausman}.

A special case of Besov spaces are Sobolev spaces. They are defined as follows
(see Nickl and P\"{o}tscher, 2007, Section 3.3.2). Let $1<p<\infty$, real
$s\geq0$ and
\[
\mathcal{H}_{p}^{s}\left(  \mathbb{R}^{d}\right)  =\left\{  f\in
\mathcal{L}_{p}\left(  \mathbb{R}^{d},\lambda\right)  :\left\Vert f\right\Vert
_{s,p,\lambda}\equiv\left\Vert F^{-1}\left(  \left\langle x\right\rangle
^{s}Ff\right)  \right\Vert _{p,\lambda}<\infty\right\}
\]
where the norms are formulated in terms of the Fourier transform $F.$ When
$s\geq0$ is integer, an equivalent (semi)norm on $\mathcal{H}_{p}^{s}\left(
\mathbb{R}^{d}\right)  $ is given by
\[
\left\Vert f\right\Vert =\sum_{0\leq\left\vert \alpha\right\vert \leq
s}\left\Vert D^{\alpha}f\right\Vert _{p,\lambda}.
\]
Similar as before define the Banach space $H_{p}^{s}\left(  \mathbb{R}%
^{d}\right)  $ of continuous functions for $s>d/p$ as
\[
H_{p}^{s}\left(  \mathbb{R}^{d}\right)  =\mathcal{H}_{p}^{s}\left(
\mathbb{R}^{d}\right)  \cap\left\{  f:\in C\left(  \mathbb{R}^{d}\right)
\right\}  .
\]
The weighted Sobolev space is given by
\[
H_{p}^{s}\left(  \mathbb{R}^{d},\vartheta\right)  =\left\{  f:f\left(
.\right)  \left\langle x\right\rangle ^{\vartheta/2}\in H_{p}^{s}\left(
\mathbb{R}^{d}\right)  \right\}  .
\]
For $s>0$, $s$ not integer, the H\"{o}lder space is defined as the space
$C^{s}\left(  \mathbb{R}^{d}\right)  $ of all $\left\lfloor s\right\rfloor
$-times differentiable functions $f$ with finite norm
\[
\left\Vert f\right\Vert _{s,\infty}=\sum_{0\leq\left\vert \alpha\right\vert
\leq\left\lfloor s\right\rfloor }\left\Vert D^{\alpha}f\right\Vert _{\infty
}+\sum_{\left\vert \alpha\right\vert =\left\lfloor s\right\rfloor }\sup_{x\neq
y}\frac{\left\vert D^{\alpha}f\left(  x\right)  -D^{\alpha}f\left(  y\right)
\right\vert }{\left\vert x-y\right\vert ^{s-\left\lfloor s\right\rfloor }}.
\]
The weighted space $C^{s}\left(  \mathbb{R}^{d},\vartheta\right)  $ is given
by
\[
C^{s}\left(  \mathbb{R}^{d},\vartheta\right)  =\left\{  f:\left\Vert f\left(
.\right)  \left\langle x\right\rangle ^{\vartheta/2}\right\Vert _{s,\infty
}<\infty\right\}  .
\]
Related is the Zygmund space $\mathcal{C}^{s}\left(  \mathbb{R}^{d}\right)  $
for $s>0$ defined in Triebel (1983, p.36) or Triebel (1992, p.4). Let
\[
\left\Vert f\right\Vert _{s,\infty}^{z}=\sum_{0\leq\left\vert \alpha
\right\vert \leq\left[  s\right]  ^{-}}\left\Vert D^{\alpha}f\right\Vert
_{\infty}+\sum_{\left\vert \alpha\right\vert =\left[  s\right]  ^{-}}%
\sup_{0\neq z\in\mathbb{R}^{d}}\left\vert z\right\vert ^{-\left\{  s\right\}
^{+}}\left\Vert \Delta_{z}^{2}D^{\alpha}f\right\Vert _{\infty}.
\]
By Triebel (1992, p.5), $\mathcal{C}^{s}\left(  \mathbb{R}^{d},\vartheta
\right)  =$ $C^{s}\left(  \mathbb{R}^{d},\vartheta\right)  $ when $s>0$ and
$s$ is not integer.

Let $\mathfrak{X\subset}\mathbb{R}^{d}$ be a bounded Borel set. The space
$C^{s}\left(  \mathfrak{X}\right)  $ is considered by van der Vaart and
Wellner (1996, p. 154) under the additional constraint that $\left\Vert
f\right\Vert _{s,\infty}\leq M$ for some bounded constant $M.$ As noted there,
when $0<s<1$, $C^{s}\left(  \mathfrak{X}\right)  $ contains the Lipschitz
functions (see Adams and Fournier 2003, Theorem 1.34).

\section{New Results\label{Sec_CLT}}

The following result gives upper bounds for entropy with bracketing on the
normed space $\mathcal{L}_{2,\beta}\left(  P\right)  .$ It extends Theorem 1
of Nickl and P\"{o}tscher (2007) to the space $\mathcal{L}_{2,\beta}\left(
P\right)  $ which plays a crucial role in obtaining a functional CLT for
dependent processes.

\begin{theorem}
\label{Theorem Bracketing Besov}Assume that $1\leq p\leq\infty$, $1\leq
q\leq\infty,$ $\vartheta\in\mathbb{R}$ and $s>d/p.$ Further assume that
$\mathcal{F\subset}B_{pq}^{s}\left(  \mathbb{R}^{d},\vartheta\right)  $ is
nonempty and bounded. If $\vartheta>0$ then
\[
H_{[]}\left(  \delta,\mathcal{F},\left\Vert {}\right\Vert _{2,\beta}\right)
\precsim\left\{
\begin{array}
[c]{cc}%
\delta^{-d/s} & \text{if }\vartheta>s-d/p\\
\delta^{-\left(  \vartheta/d+1/p\right)  ^{-1}} & \text{if }\vartheta<s-d/p
\end{array}
\right.  .
\]
If $\vartheta\leq0$ and if for some $\gamma>0$ it holds that
\[
\left\Vert \left\langle \chi_{t}\right\rangle ^{\left(  \gamma-\vartheta
\right)  /2}\right\Vert _{2,\beta}<\infty
\]
then it follows that
\[
H_{[]}\left(  \delta,\mathcal{F},\left\Vert {}\right\Vert _{2,\beta}\right)
\precsim\left\{
\begin{array}
[c]{cc}%
\delta^{-d/s} & \text{if }\gamma>s-d/p\\
\delta^{-\left(  \gamma/d+1/p\right)  ^{-1}} & \text{if }\gamma<s-d/p
\end{array}
\right.  .
\]

\end{theorem}

The difference between Nickl and P\"{o}tscher (2007, Theorem 1) and Theorem
\ref{Theorem Bracketing Besov} is that bracketing is with respect to the norm
$\left\Vert .\right\Vert _{2,\beta}$ rather than the conventional $\left\Vert
.\right\Vert _{r,P}$ norm on $\mathcal{L}_{r}\left(  \mathbb{R}^{d},P\right)
$. Theorem \ref{Theorem Bracketing Besov} directly leads to a functional CLT
based on the theory of DMR. A corollary to Theorem
\ref{Theorem Bracketing Besov} is obtained for the case when the function
space $\mathcal{F}$ is restricted to a bounded domain $\mathfrak{X.}$

\begin{corollary}
\label{Corollary Bracketing Besov}Let $\mathfrak{X\subset}\mathbb{R}^{d}$ and
there exists a finite $M$ with $\left\langle x\right\rangle \leq M$ for all
$x\in\mathfrak{X.}$ Assume that $1\leq p\leq\infty$, $1\leq q\leq\infty,$
$\vartheta\in\mathbb{R}$ and $s>d/p.$ Further assume that $\mathcal{F\subset
}B_{pq}^{s}\left(  \mathfrak{X},\vartheta\right)  $ is nonempty and bounded.
Then,
\[
H_{[]}\left(  \delta,\mathcal{F},\left\Vert {}\right\Vert _{2,\beta}\right)
\precsim\left\{
\begin{array}
[c]{cc}%
\delta^{-d/s} & \text{if }\vartheta>s-d/p\\
\delta^{-\left(  \vartheta/d+1/p\right)  ^{-1}} & \text{if }\vartheta<s-d/p
\end{array}
\right.  .
\]

\end{corollary}

The bounds on bracketing numbers obtained in Theorem
\ref{Theorem Bracketing Besov} and Corollary \ref{Corollary Bracketing Besov}
can now be applied to obtain a functional central limit theorem based on
Theorem 1 of DMR. The proof uses the tail decay properties of weighted
function spaces to establish that $\mathcal{F\subset L}_{2,\beta}\left(
P\right)  .$ This property is satisfied without further assumptions about the
marginal distribution of $\chi_{t}$ if $\vartheta>0$.

\begin{theorem}
\label{FCLT}Let $\chi_{t}$ be a strictly stationary and $\beta$-mixing
process. Assume that (\ref{Cond_beta_sum1}) holds. Assume that $1\leq
p\leq\infty$, $1\leq q\leq\infty,$ $\vartheta\in\mathbb{R}$ and $s>d/p.$
Further assume that $\mathcal{F\subset}B_{pq}^{s}\left(  \mathbb{R}%
^{d},\vartheta\right)  $ is nonempty and bounded. Assume that one of the
following conditions hold: \newline(i) $\vartheta>0,$ $\vartheta>s-d/p$ and
$s/d>1/2$; \newline(ii) $\vartheta>0,$ $\vartheta<s-d/p$ and $\vartheta
/d+1/p>1/2;$ \newline(iii) $\vartheta\leq0$ and for some $\gamma>0$ it follows
that $\left\Vert \left\langle \chi_{t}\right\rangle ^{\left(  \gamma
-\vartheta\right)  /2}\right\Vert _{2,\beta}<\infty,$ $\gamma>s-d/p$ and
$s/d>1/2;$ \newline(iv) $\vartheta\leq0$ and for some $\gamma>0$ it follows
that $\left\Vert \left\langle \chi_{t}\right\rangle ^{\left(  \gamma
-\vartheta\right)  /2}\right\Vert _{2,\beta}<\infty,$ $\gamma<s-d/p$ and
$\gamma/d+1/p>1/2.$ \newline Then, $v_{n}\left(  f\right)  \rightsquigarrow
v\left(  f\right)  $ where $v\left(  f\right)  $ is a Gaussian process with
covariance function $\Gamma$ and a.s. uniformly continuous sample paths.
\end{theorem}

Note that the conditions $1/2<s/d$ and $1/2<\gamma/d+1/p$ are the same as the
conditions given in Corollary 5 of Nickl and P\"{o}tscher (2007) for the iid
case. In the time series case these conditions need to hold in conjunction
with bounds on the $\beta$-mixing coefficients and, when $\vartheta\leq0,$ the
moment condition $\left\Vert \left\langle \chi_{t}\right\rangle ^{\left(
\gamma-\vartheta\right)  /2}\right\Vert _{2,\beta}<\infty$.

Theorem \ref{FCLT} shows that an empirical process CLT can be obtained under
the minimal Condition (\ref{Cond_beta_sum1}) if $\mathcal{F}$ is a space of
functions that asymptote to zero rapidly enough, measured by the parameter
$\vartheta>0.$ If the decay is rapid enough relative to smoothness as in case
(i) then the functional CLT holds under the minimal condition $s/d>1/2.$ Even
in case (ii) one still obtains a result with only Condition
(\ref{Cond_beta_sum1}) imposed on the dependence of the process.

When $\vartheta\leq0$ the CLT only holds under additional moment restrictions
and summability conditions for the $\beta$-mixing coefficients that are
stronger than those imposed by (\ref{Cond_beta_sum1}). The $\left\Vert
.\right\Vert _{2,\beta}$ norm provides a compact summary of these conditions
at the cost of being less easy to apply to statistical problems. It is also
harder to compare results formulated for bounds on $\left\Vert .\right\Vert
_{2,\beta}$ with results in the literature. Theorem \ref{FCLT_MB} below gives
sufficient conditions in terms of moments for $\chi_{t}$ and the summability
of mixing coefficients without directly relying on the $\left\Vert
.\right\Vert _{2,\beta}$ norm.

The results given here complement the ones in Rio (2013). If a process is
strictly stationary and $\beta$-mixing with Condition (\ref{Cond_beta_sum1})
and $f\in B_{pq}^{s}\left(  \mathbb{R}^{d},\vartheta\right)  $ with
$\vartheta>s-d/p$ then Theorem \ref{FCLT}(i) establishes a functional CLT
under the conditions that $s>d/p$ and $s/d>1/2.$ In particular, if $p=\infty,$
then the FCLT holds under the minimal condition that $\vartheta>s>0$ and
$s/d>1/2.$ This case is not covered by the results in Rio (2013). To see this
note that $B_{p_{1}\infty}^{s}\left(  \mathbb{R}^{d}\right)  \subset
B_{p_{2}\infty}^{s+d/p_{1}-d/p_{2}}\left(  \mathbb{R}^{d}\right)  $ for
$p_{1}\leq p_{2}\leq\infty$ by Triebel (1983, 2.7.1) indicating that the class
$B_{p\infty}^{s}\left(  \mathbb{R}^{d}\right)  $ for $p>2$, which is covered
by Theorem \ref{FCLT}, is a larger class than the one considered by Rio
(2013). Further, from Haroske and Triebel (1994, 2005) it follows for
$\vartheta>0,$ $\vartheta/d<1,$ $s_{1}-s_{2}>0$ and $p_{1}\left(
1-\vartheta/d\right)  <p_{2}$ that $B_{p_{1}\infty}^{s_{1}}\left(
\mathbb{R}^{d},\vartheta\right)  $ is embedded in $B_{p_{2}\infty}^{s_{2}%
}\left(  \mathbb{R}^{d}\right)  .$ For example, when $d=1$ the constraints
$s>1/2,$ $s>1/p$, $\vartheta<1$ and $p_{1}\left(  1-\vartheta\right)  <2$ must
hold for $B_{p_{1}\infty}^{s_{1}}\left(  \mathbb{R}^{d},\vartheta\right)  $ to
be embedded in $B_{p_{2}\infty}^{s_{2}}\left(  \mathbb{R}^{d}\right)  .$ Thus,
for Rio's results to encompass Theorem \ref{FCLT} one needs $p<2/\left(
1-\vartheta\right)  .$ The results of Rio (2013) then cover the spaces
$B_{p\infty}^{s}\left(  \mathbb{R}^{d},\vartheta\right)  $ for values of
$\vartheta<1$ and values of $p\leq\infty.$ However, as $\vartheta$ approaches
$0,$ the largest value $p$ can take approaches $2$ while such a constraint
does not apply to Theorem \ref{FCLT}. On the other hand, Rio (2013) covers
cases with $\vartheta=0$ and $p\leq2$ which can only be handled by Theorem
\ref{FCLT} under additional moment restrictions and stronger assumptions on
the $\beta$-mixing coefficients.

When $p=2,$ then $s/d>1/2$ and $\vartheta>s-d/2$ lead to a FCLT by means of
Theorem \ref{FCLT}. This case essentially corresponds to Rio (2013) when
$s-d/2$ is close to zero. By Triebel (1983, 2.7.1) it follows that
$B_{pq}^{s_{1}}\left(  \mathbb{R}^{d},\vartheta\right)  $ is continuously
embedded in $B_{pq}^{s_{0}}\left(  \mathbb{R}^{d},\vartheta\right)  $ for
$s_{1}\geq s_{0}.$ Thus, to apply Theorem \ref{FCLT} one can always choose $s$
small enough such that $s-d/2$ is arbitrarily small and therefore $\vartheta$
can be chosen small. If $\vartheta<s-d/p$ then Theorem \ref{FCLT}(ii) holds
under the condition that $\vartheta/d>1/2+1/p$ such that the CLT holds for $p$
sufficiently large and $s/d>1/2.$

These arguments indicate that the results in Rio (2013) are slightly sharper
for the case when $p\in\left[  1,2\right]  $ because of the requirement in
Theorem \ref{FCLT} that $\vartheta>s-d/p.$ In addition, by Triebel (1983,
2.3.2, Proposition 2), $B_{pq_{0}}^{s}\left(  \mathbb{R}^{d},\vartheta\right)
$ is continuously embedded in $B_{pq_{1}}^{s}\left(  \mathbb{R}^{d}%
,\vartheta\right)  $ for $q_{0}\leq q_{1}\leq\infty$ and $p>0$ such that
$B_{pq}^{s}\left(  \mathbb{R}^{d}\right)  $ is continuously embedded in
$Lip^{\ast}\left(  s,p,\mathbb{R}^{d}\right)  .$ This implies that the results
in Rio cover the spaces $B_{pq}^{s}\left(  \mathbb{R}^{d}\right)  $ for
$p\in\left[  1,2\right]  $ and $q\leq\infty.$

In summary, the results in Theorem \ref{FCLT} are very similar to Rio (2013)
when $p\leq2$ and the tail behavior of the function class is controlled by a
polynomial. However, the results are achieved with simpler proofs. Because of
the embedding result in Triebel (1983, 2.7.1), additional function classes are
covered by Theorem \ref{FCLT} that are not contained in Rio (2013) when $p>2.$
Theorem \ref{FCLT} also covers cases when $\vartheta\leq0$ and $p\leq\infty$
that are not covered by Rio (2013). However, in these situations somewhat
stronger assumptions than (\ref{Cond_beta_sum1}) need to be imposed on
dependence. Here the case $\vartheta=0$ and $p=\infty$ may be of particular
interest since the tail behavior of $f\left(  x\right)  $ no longer
necessarily satisfies $\lim_{\left\Vert x\right\Vert }f\left(  x\right)
\rightarrow0$ (see Proposition 3 of NP). This is one example of a case not
covered by the results in Rio (2013).

Another result that is not directly covered by Theorem \ref{FCLT} is Rio
(2013, Theorem 8.1). Rio considers the generalized Lipschitz spaces
$Lip^{\ast}\left(  s,p,\mathbb{R}^{d}\right)  $ defined in Meyer (1992). Meyer
(1992, Proposition 7, p. 200) shows that every $f\in Lip^{\ast}\left(
s,p,\mathbb{R}^{d}\right)  $ is in $\mathcal{B}_{p\infty}^{s}\left(
\mathbb{R}^{d}\right)  .$ Rio (2013, Proposition 8.1) gives an equivalent norm
$\left\Vert f\right\Vert _{ond}$ for functions $f\in Lip^{\ast}\left(
s,p,\mathbb{R}^{d}\right)  .$ Rio shows that for every strongly mixing and
stationary sequence with $\sum_{m=1}^{\infty}\alpha_{m}<\infty,$ $f\in
Lip^{\ast}\left(  s,p,\mathbb{R}^{d}\right)  $ with $p\in\left[  1,2\right]
,$ $s>d/p$ and $\left\Vert f\right\Vert _{ond}\leq a$ for some constant
$a<\infty,$ the empirical process $v_{n}\left(  f\right)  $ satisfies a
stochastic equicontinuity condition and thus a functional central limit
theorem. Rio (2013, Theorem 8.1) is not covered by the theory in this paper
because the concept of strongly mixing sequences is slightly weaker than
$\beta$-mixing.

An immediate corollary to Theorem \ref{FCLT} obtains for the case where
$\chi_{t}$ takes values in a bounded set $\mathfrak{X\subset}\mathbb{R}^{d}$.

\begin{corollary}
\label{Corollary FCLT}Let $\chi_{t}$ be strictly stationary and $\beta
$-mixing. Assume that $P\left(  \chi_{t}\in\mathfrak{X}\right)  =1$ for a
bounded Borel set $\mathfrak{X\subset}\mathbb{R}^{d}$ and there exists a
finite $M$ with $\left\langle x\right\rangle \leq M$ for all $x\in
\mathfrak{X.}$ Assume that (\ref{Cond_beta_sum1}) holds. Assume that $1\leq
p\leq\infty$, $1\leq q\leq\infty,$ $\vartheta\in\mathbb{R}$ and for
$s,d<\infty$ and $s>d/p.$ Further assume that $\mathcal{F\subset}B_{pq}%
^{s}\left(  \mathfrak{X}\right)  $ is nonempty and bounded. Assume that
$s/d>1/2$. Then, $v_{n}\left(  f\right)  \rightsquigarrow v\left(  f\right)  $
where $v\left(  f\right)  $ is a Gaussian process with covariance function
$\Gamma$ and a.s. uniformly continuous sample paths.
\end{corollary}

Corollary \ref{Corollary FCLT} show that for smooth function classes
restricted to a bounded set a functional CLT holds under the minimal
dependence condition (\ref{Cond_beta_sum1}).

When the asymptotic behavior of $f$ as $\left\Vert x\right\Vert \rightarrow
\infty$ is proportional to $\left\langle \chi_{t}\right\rangle ^{-\vartheta
/2}$ and $\vartheta\leq0,$ then more restrictive conditions on the dependence
need to be imposed. This happens implicitly through the condition
\begin{equation}
\left\Vert \left\langle \chi_{t}\right\rangle ^{\left(  \gamma-\vartheta
\right)  /2}\right\Vert _{2,\beta}<\infty\label{Moment_Bound}%
\end{equation}
which must hold for some $\gamma>0.$ The advantage of this condition is that
it only involves the marginal distribution of $\chi_{t}$ and not the
properties of the functional class, other than through the parameter
$\vartheta$. Results in DMR can be used to give simple sufficient conditions
for \ref{Moment_Bound}. Under additional assumptions about the order of
$\beta_{m}$ and moment restrictions on the marginal distribution of
$\left\Vert \chi_{t}\right\Vert ^{2}$ the following result can be given for
the case when $\vartheta\leq0,$ i.e. when $\lim_{\left\Vert x\right\Vert
}f\left(  x\right)  \rightarrow0$ does not necessarily hold.

\begin{theorem}
\label{FCLT_MB}Let $\chi_{t}$ be strictly stationary and $\beta$-mixing.
Assume that for some $r>1,$ $\sum_{m=1}^{\infty}m^{1/\left(  r-1\right)
}\beta_{m}<\infty$ holds. Assume that $1\leq p\leq\infty$, $1\leq q\leq
\infty,$ $\vartheta\in\mathbb{R}$, $\vartheta\leq0$ and $s>d/p.$ Further
assume that $\mathcal{F\subset}B_{pq}^{s}\left(  \mathbb{R}^{d},\vartheta
\right)  $ is nonempty and bounded. Assume that for some $\gamma>0$ such that
$r\left(  \gamma-\vartheta\right)  >1$ it holds that $E\left[  \left\Vert
\chi_{t}\right\Vert ^{2r\left(  \gamma-\vartheta\right)  }\right]  <\infty$
and that either \newline(i) $\gamma>s-d/p$ and $s/d>1/2$ or \newline(ii)
$\gamma<s-d/p$ and $\gamma/d+1/p>1/2.$\newline Then, $v_{n}\left(  f\right)
\rightsquigarrow v\left(  f\right)  $ where $v\left(  f\right)  $ is a
Gaussian process with covariance function $\Gamma$ and a.s. uniformly
continuous sample paths.
\end{theorem}

The form of the last theorem is particularly useful when a comparison with
other results in the literature is desired, since those results are often
presented in terms of separate moment bounds and size restrictions on mixing coefficients.

More generally, the results show that in weighted Besov spaces control over
tail behavior of the function class can be utilized to give sufficient
conditions for a CLT that directly involves the marginal distribution of
$\chi_{t}$ rather than that of $f\left(  \chi_{t}\right)  .$ This is possible
because the asymptotic behavior of $f\left(  \chi_{t}\right)  $ is controlled
by terms that are functions of $\left\Vert \chi_{t}\right\Vert .$ The next
corollary gives explicit versions of the previous general results for Sobolev,
H\"{o}lder and Lipschitz classes of functions.

The following Corollary is a special case of Theorem \ref{FCLT}. The proof
follows in the same way as the proofs of similar corollaries in Nickl and
P\"{o}tscher (2007) by arguing that bounded subsets of $H_{p}^{s}\left(
\mathbb{R}^{d},\vartheta\right)  $ are also bounded subsets of $B_{p\infty
}^{s}\left(  \mathbb{R}^{d},\vartheta\right)  .$

\begin{corollary}
Let $\chi_{t}$ be a strictly stationary and $\beta$-mixing process. Assume
that (\ref{Cond_beta_sum1}) holds. Assume that $1<p\leq\infty$, $\vartheta
\in\mathbb{R}$ and $s>d/p.$ Further assume that $\mathcal{F\subset}H_{p}%
^{s}\left(  \mathbb{R}^{d},\vartheta\right)  $ is nonempty and bounded. Assume
that one of the following conditions hold: \newline(i) $\vartheta>0,$
$\vartheta>s-d/p$ and $s/d>1/2$; \newline(ii) $\vartheta>0,$ $\vartheta<s-d/p$
and $\vartheta/d+1/p>1/2;$ \newline(iii) $\vartheta\leq0$ and for some
$\gamma>0$ it follows that $\left\Vert \left\langle \chi_{t}\right\rangle
^{\left(  \gamma-\vartheta\right)  /2}\right\Vert _{2,\beta}<\infty,$
$\gamma>s-d/p$ and $s/d>1/2;$ \newline(iv) $\vartheta\leq0$ and for some
$\gamma>0$ it follows that $\left\Vert \left\langle \chi_{t}\right\rangle
^{\left(  \gamma-\vartheta\right)  /2}\right\Vert _{2,\beta}<\infty,$
$\gamma<s-d/p$ and $\gamma/d+1/p>1/2.$ \newline Then, $v_{n}\left(  f\right)
\rightsquigarrow v\left(  f\right)  $ where $v\left(  f\right)  $ is a
Gaussian process with covariance function $\Gamma$ and a.s. uniformly
continuous sample paths.
\end{corollary}

The following corollary again considers the special case where the domain of
the function space is a bounded subset of $\mathbb{R}^{d}.$

\begin{corollary}
\label{Corollary_Besov_Bounded}Let $\chi_{t}$ be a strictly stationary and
$\beta$-mixing process. Assume that $P\left(  \chi_{t}\in\mathfrak{X}\right)
=1$ where $\mathfrak{X\subset}\mathbb{R}^{d}$ and there exists a finite $M$
with $\left\langle x\right\rangle \leq M$ for all $x\in\mathfrak{X.}$ Assume
that (\ref{Cond_beta_sum1}) holds. Assume that $1<p\leq\infty$, $\vartheta
\in\mathbb{R}$ and $s,d<\infty$ with $s>d/p.$ Further assume that
$\mathcal{F\subset}H_{p}^{s}\left(  \mathfrak{X,}\vartheta\right)  $ is
nonempty and bounded. Assume that $s/d>1/2$. Then, $v_{n}\left(  f\right)
\rightsquigarrow v\left(  f\right)  $ where $v\left(  f\right)  $ is a
Gaussian process with covariance function $\Gamma$ and a.s. uniformly
continuous sample paths.
\end{corollary}

Andrews (1991) considers the space $\mathcal{H}_{p}^{s}\left(  \mathfrak{X}%
\right)  $ where $\mathfrak{X}$ is a bounded subset of $\mathbb{R}^{d}.$ He
allows for heterogeneous near epoch dependent processes which include as
special cases strong mixing stationary sequences. Since $\beta$-mixing implies
strong mixing the results of this paper are obtained under somewhat stronger
assumptions as far as the mixing concept and stationarity requirements are
concerned. On the other hand, no boundedness of $\mathfrak{X}$ is required.
Andrews (1991, p.199) discusses some ways of relaxing the boundedness
assumption regarding the support but does not provide a general treatment.
Moreover, as pointed out by Nickl and P\"{o}tscher (2007, p. 179 and p. 196)
it follows for $f\in$ $H_{p}^{s}\left(  \mathbb{R}^{d}\right)  $,
$\lim_{\left\Vert x\right\Vert }f\left(  x\right)  \rightarrow0$ while this is
not necessarily the case for $f\in H_{p}^{s}\left(  \mathbb{R}^{d}%
,\vartheta\right)  $ and $\vartheta<0.$

Andrews (1991, Theorem 4 and Comment 1) obtains a functional central limit
theorem for strong mixing processes of size $-2,$ $f\in\mathcal{H}_{2}%
^{s}\left(  \mathfrak{X}\right)  $ and $s/d>1/2.$ Corollary
\ref{Corollary_Besov_Bounded} shows that, at least under the additional
assumption of stationarity and $\beta$-mixing but only satisfying
(\ref{Cond_beta_sum1}), this result can be obtained for all functions in
$H_{p}^{s}\left(  \mathfrak{X}\right)  $ with $s/d>1/2$. Note that a $\beta
$-mixing process that satisfies Condition (\ref{Cond_beta_sum1}) also is
$\alpha$-mixing with $\sum_{m=1}^{\infty}\alpha_{m}<\infty$ but is not
necessarily $\alpha$-mixing of size $-2.$ In this sense, the conditions given
here are complementary to Andrews (1991).

The following corollaries specialize previous results to H\"{o}lder spaces.

\begin{corollary}
Let $\chi_{t}$ be strictly stationary and $\beta$-mixing. Assume that
(\ref{Cond_beta_sum1}) holds. Assume that $\vartheta\in\mathbb{R}$ and
$s>d/2.$ Further assume that $\mathcal{F\subset}C^{s}\left(  \mathbb{R}%
^{d},\vartheta\right)  $ is nonempty and bounded. Assume that one of the
following conditions hold: \newline(i) $\vartheta>0,$ $\vartheta>s$ and
$s/d>1/2$; \newline(ii) $\vartheta>0,$ $\vartheta<s$ and $\vartheta/d>1/2;$
\newline(iii) $\vartheta\leq0$ and for some $\gamma>0$ it follows that
$\left\Vert \left\langle \chi_{t}\right\rangle ^{\left(  \gamma-\vartheta
\right)  /2}\right\Vert _{2,\beta}<\infty,$ $\gamma>s$ and $s/d>1/2;$
\newline(iv) $\vartheta\leq0$ and for some $\gamma>0$ it follows that
$\left\Vert \left\langle \chi_{t}\right\rangle ^{\left(  \gamma-\vartheta
\right)  /2}\right\Vert _{2,\beta}<\infty,$ $\gamma<s$ and $\gamma/d>1/2.$
\newline Then, $v_{n}\left(  f\right)  \rightsquigarrow v\left(  f\right)  $
where $v\left(  f\right)  $ is a Gaussian process with covariance function
$\Gamma$ and a.s. uniformly continuous sample paths.
\end{corollary}

The proof follows again from noting that $\mathcal{F}$ is a bounded subset in
$B_{\infty\infty}^{s}\left(  \mathbb{R}^{d},\vartheta\right)  ,$ see Nickl and
P\"{o}tscher (2007, p. 188). As before, additional results for the cases of
bounded support can be stated as follows.

\begin{corollary}
\label{Corollary_Holder_Bounded}Let $\chi_{t}$ be a strictly stationary and
$\beta$-mixing. Assume that $P\left(  \chi_{t}\in\mathfrak{X}\right)  =1$
where $\mathfrak{X\subset}\mathbb{R}^{d}$ and there exists a finite $M$ with
$\left\langle x\right\rangle \leq M$ for all $x\in\mathfrak{X.}$ Assume that
(\ref{Cond_beta_sum1}) holds. Assume that $\vartheta\in\mathbb{R}$ ,
$s,d<\infty$ and $s>0.$ Further assume that $\mathcal{F\subset}C^{s}\left(
\mathfrak{X},\vartheta\right)  $ is nonempty and bounded. Assume that
$s/d>1/2$. Then, $v_{n}\left(  f\right)  \rightsquigarrow v\left(  f\right)  $
where $v\left(  f\right)  $ is a Gaussian process with covariance function
$\Gamma$ and a.s. uniformly continuous sample paths.
\end{corollary}

Andrews (1991, Comment 3) also considers the case of strong mixing processes
of size $-2$ and Lipschitz function classes. More specifically, when
$\mathfrak{X}$ is a bounded interval on $\mathbb{R}$, a functional central
limit theorem holds for functions $f$ such that $\left\vert f\left(  x\right)
-f\left(  y\right)  \right\vert \leq K\left\vert x-y\right\vert ^{s}$ with
$s\in(1/2,1].$ By Adams and Fournier (2003, Theorem 1.34) the function class
$C^{s}\left(  \mathfrak{X}\right)  $ with $s\in(1/2,1)$ contains the Lipschitz
functions with $s\in(1/2,1)$. Then, Corollary \ref{Corollary_Holder_Bounded}
can be used to establish a functional central limit theorem for Lipschitz
functions and for stationary $\beta$-mixing processes that satisfy Condition
(\ref{Cond_beta_sum1}). Note that when $s\in(1/2,1)$ and $\mathfrak{X}$ is a
bounded interval, it follows that for $d=1$ the condition $s/d>1/2$ is satisfied.

When $\vartheta\leq0$ such that $\lim_{\left\Vert x\right\Vert }f\left(
x\right)  \rightarrow0$ does not hold, a more specific result can be given for
functions in $C^{s}\left(  \mathbb{R}^{d},\vartheta\right)  $ as long as one
is willing to impose additional conditions on the rate of decay of $\beta_{m}%
$. This is done in the following corollary.

\begin{corollary}
\label{Holder_Unbounded_Moment}Let $\chi_{t}$ be strictly stationary and
$\beta$-mixing. Assume that for some $r>1,$ $\sum_{m=1}^{\infty}m^{1/\left(
r-1\right)  }\beta_{m}<\infty$ holds. Assume that $\vartheta\in\mathbb{R}$,
$\vartheta\leq0$ and $s>0.$ Further assume that $\mathcal{F\subset}%
C^{s}\left(  \mathbb{R}^{d},\vartheta\right)  $ is nonempty and bounded.
Assume that for some $\gamma>0$ such that $r\left(  \gamma-\vartheta\right)
>1$ it holds that that $E\left[  \left\Vert \chi_{t}\right\Vert ^{2r\left(
\gamma-\vartheta\right)  }\right]  <\infty$ and that either \newline(i)
$\gamma>s$ and $s/d>1/2$ or \newline(ii) $\gamma<s$ and $\gamma/d>1/2.$
\newline Then, $v_{n}\left(  f\right)  \rightsquigarrow v\left(  f\right)  $
where $v\left(  f\right)  $ is a Gaussian process with covariance function
$\Gamma$ and a.s. uniformly continuous sample paths.
\end{corollary}

Corollary \ref{Holder_Unbounded_Moment} should only be applied to cases where
$\vartheta\leq0.$ As for previous results, when $\vartheta>0$, the functional
central limit theorem can be established under weaker assumptions.

The results in DMR are stated in general terms and form the basis for what is
derived here. Nevertheless, on p.403-405 DMR provide a number of different
approaches that can be used to replace high level assumptions with more
primitive conditions. These methods do not lead to the sharpest possible
results as far as conditions on $\beta_{m}$ are concerned for the classes of
functions considered by Rio (2013). For functions whose tail decay is well
controlled by a polynomial or for functions that are restricted to a bounded
domain Theorem \ref{FCLT} also delivers sharper results. In particular,
Theorem \ref{FCLT} shows that $\vartheta>0,$ i.e. when tail behavior is
controlled by polynomials, the functional CLT can be obtained without
requiring the additional moment bound in (\ref{Moment_Bound}). As a result,
neither the marginal distribution of $\chi_{t}$ nor the dependence of the
process need further restrictions beyond Condition (\ref{Cond_beta_sum1}). On
the other hand, the results in DMR lead to similar conditions as the ones
given in Theorem \ref{FCLT_MB} for spaces where $\vartheta\leq0.$ The
following result illustrates this. By exploiting condition (2.11) in DMR and
applying Theorem 1 in Nickl and P\"{o}tscher (2007) one obtains the following.

\begin{theorem}
\label{FCLT2}Let $1\leq p\leq\infty,$ $1\leq q\leq\infty,$ $\vartheta
\in\mathbb{R}$ and $s-d/p>0.$ For $1<r<\infty$ let $\chi_{t}$ be a strictly
stationary, absolutely regular process such that $\sum_{m=1}^{\infty
}m^{1/\left(  r-1\right)  }\beta_{m}<\infty.$ Assume that for some $\gamma>0$
such that $r\left(  \gamma-\vartheta\right)  >1$ the moment bound
\begin{equation}
\left\Vert \left\langle \chi_{t}\right\rangle ^{\left(  \gamma-\vartheta
\right)  /2}\right\Vert _{2r,P}<\infty\label{Moment Cond FCLT2}%
\end{equation}
holds. Let $\mathcal{F}$ be a bounded subset of $B_{pq}^{s}\left(
\mathbb{R}^{d},\vartheta\right)  .$ Furthermore one of the conditions holds:
\newline i) $\gamma>s-d/p$ and $1/2<s/d$ \newline ii) $\gamma<s-d/p$ and
$1/2<\gamma/d+1/p.$ \newline Then $v_{n}\left(  f\right)  \rightsquigarrow
v\left(  f\right)  $ where $v\left(  f\right)  $ is a Gaussian process with
covariance function $\Gamma$ and a.s. uniformly continuous sample paths.
\end{theorem}

The conditions of Theorem \ref{FCLT2} are the same as given in Theorem
\ref{FCLT_MB} for the case when $\vartheta\leq0.$ However, the limitation of
Theorem \ref{FCLT2} over Theorems \ref{FCLT} and \ref{FCLT_MB} is that it does
not deliver a functional central limit theorem under the minimal condition
(\ref{Cond_beta_sum1}) when $\vartheta>0.$

\section{Application: A Hausman Test for Linearity\label{Sec_Hausman}}

This section considers the problem of testing the specification of the
conditional mean $g\left(  x\right)  =E\left[  y|x\right]  $ for a process
$\chi_{t}=\left(  y_{t},x_{t}\right)  $. The purpose of the section is to
illustrate how the central limit theory developed in this paper can be used to
obtain limiting results for fairly general classes of processes and
conditional mean functions. Because unbounded domains are important in time
series applications, the theory for weighted function spaces is particularly
relevant. Minimal dependence conditions in (\ref{Cond_beta_sum1}) could be
obtained under the additional assumption that the domain of $\chi_{t}$ is
bounded. This is an immediate consequence of results in earlier sections and
is only noted in passing.

The insights underlying the Hausman (1978) test are ingenious and have found
applications to a large number of testing problems in econometrics. For the
particular case considered in this paper the idea is to estimate the
conditional mean by a linear regression of $y_{t}$ on $x_{t}$. The estimator
is generally not consistent for the average partial derivative of the
conditional mean function if the conditional expectation is non-linear. An
alternative estimator uses sieve basis functions to non-parametrically
estimate the possibly non-linear regression. The average derivative of this
estimator is consistent even if the conditional expectation is non-linear.
Thus, under the null of linearity, both estimators should converge to the same
parameter. Under the alternative only the second estimator is consistent while
the first estimator will be asymptotically biased under local alternatives.
The Hausman test exploits these differences in asymptotic behavior by looking
at the difference between the two estimators. Under the null, the test
statistic has a well defined limiting distribution, while under alternatives
the difference between the estimators persists, thus lending power to the test.

Comparing two competing estimators for alternative specifications of average
partial derivatives is appealing from an applied perspective. The test
directly answers the question of whether it is worthwhile to employ more
sophisticated procedures for the estimation of average partial effects or if a
simple linear regression approach is sufficient.

There is a large literature in econometrics and statistics on specification
testing for the conditional mean. Tests against specific alternatives were
considered for example by Cox (1961), Quandt (1974) and Davidson and McKinnon
(1981). Ramsey (1969) and Newey (1985) consider tests of the orthogonality
condition in a regression model while Hausman (1978) and White (1981) consider
model specification tests based on the comparison of two estimators.
Nonparametric tests which have power against a wider range of alternatives
include Bierens (1982), Wooldridge (1992), Yatchew (1992), Zheng (1996) and
Fan and Li (1996). Bierens (1982, 1987) points out that the tests of Hausman
(1978) and White (1981) have power and in some cases consistency properties
that depend on the choice of the estimator that is consistent under both the
null and the alternative. The test considered in this section is pointwise
consistent against all fixed non-parametric deviations $h$ in the class
$B_{\infty\infty}^{s}\left(  \mathbb{R}^{d},\vartheta\right)  $ with
$\operatorname*{Cov}\left(  h,x\right)  +\pi\left(  h\right)  \neq0.$ The
parameter $\pi\left(  h\right)  $ captures the discrepancy between average
partial effects when the model is linear and when it is non-linear. The term
$\operatorname*{Cov}\left(  h,x\right)  $ accounts for linear regression bias
under the alternative. Under the null of a linear conditional mean the local
deviation $h$ is zero and $\pi\left(  h\right)  =0.$

The test proposed in this study has non-trivial power against local
alternatives of the form $n^{-1/2}h\left(  x\right)  $ for fixed $h(x)\in
B_{\infty\infty}^{s}\left(  \mathbb{R}^{d},\vartheta\right)  $ with
$\operatorname*{Cov}\left(  h,x\right)  +\pi\left(  h\right)  \neq0.$ Horowitz
and Spokoiny (2001) point out that the tests of Bierens (1982), Andrews (1997)
and Bierens and Ploberger (1997) have non-trivial power against such
alternatives while the tests of Wooldridge (1992), Yatchew (1992), Zheng
(1996) and Fan and Li (1996) only have non-trivial power against alternatives
that are local at rates slower than $n^{-1/2}.$ Horowitz and Spokoiny (2001)
develop tests that have power against more general alternatives $n^{-1/2}%
h_{n}\left(  x\right)  $ where $h_{n}\left(  x\right)  $ is a sequence of
functions. Their tests have power uniformly against certain smooth
alternatives against which the test in this paper and the tests of Bierens
(1982), Andrews (1997) and Bierens and Ploberger (1997) do not have
non-trivial power. Nevertheless, the appeal of the test proposed in this paper
is its simplicity in terms of implementation and interpretation.

The estimation problem considered in this study is semi-parametric in nature.
The distribution of the test statistic depends on the non-parametric
functional estimated by the second estimator. The influence function of the
test statistic defines an empirical process that can be used to obtain the
limiting distribution under the null and under local alternatives. This is now formalized.

Let $\chi_{t}=\left(  y_{t},x_{t}\right)  \in\mathbb{R}^{2}$ be a strictly
stationary $\beta$-mixing process and define $g\left(  x_{t}\right)  =E\left[
y_{t}|x_{t}\right]  $. Extensions to multivariate $x_{t}$ are straight forward
but omitted for ease of exposition. Consider testing the hypothesis that
$g\left(  x\right)  =\psi_{0}+\psi_{1}x$ against the alternative that
$g\left(  x\right)  $ is a non-linear function of $x.$ A linear regression
estimator for $\psi_{1}$ is generally inconsistent if $g\left(  x\right)
\neq\psi_{0}+\psi_{1}x.$ A Hausman test is then based on the squared
difference for two estimators of $E\left[  \partial g\left(  x_{t}\right)
/\partial x\right]  .$ Under the null, the average partial effect is simply
$\psi_{1}$ which is estimated as a regression of $y_{t}$ on a constant and
$x_{t}.$ Under the alternative, $E\left[  \partial g\left(  x_{t}\right)
/\partial x\right]  $ is estimated by a plug-in series estimator for $g\left(
x\right)  $.

Define $P^{\kappa}\left(  z\right)  =\left(  p_{1\kappa}\left(  z\right)
,...,p_{\kappa\kappa}\left(  z\right)  \right)  ^{\prime}$, where $p_{1\kappa
}\left(  z\right)  =z$ for all $\kappa,$ $\mu_{P}^{\kappa}=E\left[  P^{\kappa
}\left(  z_{t}\right)  \right]  $ and $\tilde{P}^{\kappa}\left(  z\right)
=P^{\kappa}\left(  z\right)  -\mu_{P}^{\kappa}.$ Define $P=\left[  P^{\kappa
}\left(  x_{1}\right)  ,...,P^{\kappa}\left(  x_{n}\right)  \right]  ^{\prime
},$
\[
MP=\left[  P^{\kappa}\left(  x_{1}\right)  -\bar{P}^{\kappa},...,P^{\kappa
}\left(  x_{n}\right)  -\bar{P}^{\kappa}\right]  ^{\prime}%
\]
where $M=I_{n}-n^{-1}\mathbf{1}_{n}\mathbf{1}_{n}^{\prime}$ with
$\mathbf{1}_{n}$ a vector of length one composed of the element one and
$\bar{P}^{\kappa}=n^{-1}\sum_{t=1}^{n}P^{\kappa}\left(  x_{t}\right)  .$ The
series estimator for $E\left[  y|x\right]  $ is $\hat{g}_{\kappa}\left(
x\right)  =\hat{\psi}_{0,\kappa}+P^{\kappa}\left(  x\right)  \hat{\psi
}_{\kappa}$ where $\hat{\psi}_{\kappa}=\left(  P^{\prime}MP\right)
^{-1}P^{\prime}My.$ The estimator for the constant is given by $\hat{\psi
}_{0,\kappa}=\bar{y}-\bar{P}^{\kappa}\hat{\psi}_{\kappa}$ with $\bar{y}%
=n^{-1}\sum_{t=1}^{n}y_{t}.$

Let $\theta=\left(  \theta_{l},\theta_{nl}\right)  $ where $\theta_{l}$ is the
average partial effect under the linear specification and $\theta
_{nl}=E\left[  \partial g\left(  x_{t}\right)  /\partial x\right]  $ is the
average partial effect under the non-linear specification. An estimator for
$\theta$ is based on a Z-estimator\footnote{This terminology appreas for
example in van der Vaart (1998, p 41).} using a plug in non-parametric
estimate $\hat{g}_{k}=\hat{g}_{k}\left(  x\right)  .$ For this purpose define
the moment function
\begin{equation}
\hat{m}\left(  \chi_{t},\theta,\hat{g}_{\kappa}\right)  =\left[
\begin{array}
[c]{c}%
\left(  y_{t}-\bar{y}-\theta_{l}\left(  x_{t}-\bar{x}\right)  \right)  \left(
x_{t}-\bar{x}\right) \\
\frac{\partial P^{\kappa}\left(  x_{t}\right)  }{\partial x}^{\prime}\hat
{\psi}_{\kappa}-\theta_{nl}%
\end{array}
\right]  \label{m_hat_function}%
\end{equation}
and let
\begin{equation}
m_{n}\left(  \theta\right)  =n^{-1}\sum_{t=1}^{n}\hat{m}\left(  \chi
_{t},\theta,\hat{g}_{\kappa}\right)  . \label{Hausman_EmpMom}%
\end{equation}
The Z-estimator $\hat{\theta}_{\kappa}=\left(  \hat{\theta}_{l},\hat{\theta
}_{nl}\right)  $ is obtained by solving $m_{n}\left(  \hat{\theta}_{\kappa
}\right)  =0.$ A Hausman test of linearity then compares the two estimators by
forming the test statistic
\[
\left(  \hat{\theta}_{l}-\hat{\theta}_{nl}\right)  ^{2}%
/\widehat{\operatorname*{Var}}\left(  \hat{\theta}_{l}-\hat{\theta}%
_{nl}\right)  .
\]
The estimator $\hat{\theta}_{l}$ is not usually efficient under the null.
However, it is well known that the Hausman testing principle can still be
applied, albeit at the cost of requiring more complicated expressions for
$\operatorname*{Var}\left(  \hat{\theta}_{l}-\hat{\theta}_{nl}\right)  .$ The
limiting distribution of $\hat{\theta}_{l}-\hat{\theta}_{nl}$ can be analyzed
within the framework of Newey (1994). The results of Newey (1994) show that
non-parametric estimation of $g\left(  x\right)  $ does affect the limiting
distribution of $\hat{\theta}_{l}-\hat{\theta}_{nl},$ but in ways that do not
depend on the specific form of the estimator for $g\left(  x\right)  .$

The limiting distribution of the test statistic is analyzed for the following
data-generating mechanism under local alternatives $g_{h}\left(  x\right)  $,%
\begin{equation}
y_{t}=\psi_{0}+\psi_{1}x_{t}+\frac{h\left(  x_{t}\right)  }{\sqrt{n}}+u_{t}
\label{DGP}%
\end{equation}
where $g_{h}\left(  x\right)  =\psi_{0}+\psi_{1}x_{t}+n^{-1/2}h\left(
x_{t}\right)  $ and $u_{t}=y_{t}-E\left[  y_{t}|x_{t}\right]  $ is such that
$E\left[  u_{t}|x_{t}\right]  =0.$ Assume that $h\left(  x\right)  =h\in
B_{\infty\infty}^{s+1}\left(  \mathbb{R},\vartheta\right)  $ for some $s>1/2$
and some $\vartheta\in\mathbb{R}.$ Let $E\left[  x_{t}\right]  =\mu_{x}$ and
set
\begin{equation}
\tilde{b}\left(  h\right)  =[b\left(  h\right)  ,0]^{\prime} \label{bg}%
\end{equation}
with $b\left(  h\right)  =E\left[  \left(  x_{t}-\mu_{x}\right)  h\left(
x_{t}\right)  \right]  .$ The term $b\left(  h\right)  $ captures biases in
estimating $\psi_{1}$ with a linear regression when $h\neq0.$ Under the null
of a linear conditional mean the function $h$ is $h_{0}\left(  x_{t}\right)
=0$ which implies that $b\left(  h_{0}\right)  =0.$ Let
\[
Q=E\left[  \partial m\left(  \chi_{t},\theta,g_{h}\right)  /\partial
\theta\right]  =\left[
\begin{array}
[c]{cc}%
\sigma_{x}^{2} & 0\\
0 & 1
\end{array}
\right]
\]
where $\sigma_{x}^{2}=E\left[  \left(  x_{t}-\mu_{x}\right)  ^{2}\right]  $.
Let $m\left(  \chi_{t},\theta,g_{h}\right)  $ be the population analog of
$\hat{m}\left(  \chi_{t},\theta,g_{h}\right)  $ defined in
(\ref{m_hat_function}) where in $m(.)$ the empirical estimates $\bar{x}$ and
$\bar{y}$ are replaced with $\mu_{x}$ and $\mu_{y}$. Let $\theta_{0}=\left(
\psi_{1},\theta_{nl}\right)  ^{\prime}$ be the value of $\theta$ for the true
data generating process (\ref{DGP}) under local alternatives. Under regularity
conditions it follows from arguments similar to Newey (1994) that for $h$
fixed,
\[
\sqrt{n}\left(  \hat{\theta}_{\kappa}-\theta_{0}\right)  =Q^{-1}\left(
n^{-1/2}\sum_{t=1}^{n}\left(  m\left(  \chi_{t},\theta_{0},g_{h}\right)
+\gamma\left(  \chi_{t}\right)  \right)  \right)  +o_{p}\left(  1\right)  .
\]
The correction term $\gamma\left(  \chi_{t}\right)  $ accounts for
non-parametric estimation of the nuisance parameter $g_{h}$ and can be derived
using the methods developed in Newey (1994). It is given by
\[
\gamma\left(  \chi_{t}\right)  =\left[
\begin{array}
[c]{c}%
0\\
\delta_{nl}\left(  x_{t}\right)
\end{array}
\right]  \left(  y_{t}-\psi_{0}-\psi_{1}x_{t}-\frac{h\left(  x_{t}\right)
}{\sqrt{n}}\right)
\]
where $\delta_{nl}\left(  x_{t}\right)  =-\zeta_{x}\left(  x\right)
^{-1}\partial\zeta_{x}\left(  x\right)  /\partial x$ and $\zeta_{x}\left(
x\right)  $ is the marginal density of $x_{t},$ see Newey (1994, p.1362) or
Hardle and Stoker (1989). Define the empirical process
\begin{equation}
v_{n}\left(  h\right)  =n^{-1/2}\sum_{t=1}^{n}\left(  m\left(  \chi_{t}%
,\theta_{0},g_{h}\right)  +\gamma\left(  \chi_{t}\right)  -E\left[  m\left(
\chi_{t},\theta_{0},g_{h}\right)  \right]  \right)  \label{Hausman_EP}%
\end{equation}

The central limit theorems developed in the first part of the paper play a
dual role in analyzing the limiting properties of $\hat{\theta}_{\kappa}.$ On
the one hand, stochastic equicontinuity properties of the empirical process
(\ref{Hausman_EP}) can be used to verify regularity conditions in Newey
(1994). On the other hand, the functional central limit theorem delivers a
stochastic process representation of the limiting distribution of $\hat
{\theta}_{\kappa}$ over the class of local alternatives.

\begin{condition}
\label{H_C1}Let $\chi_{t}$ be a strictly stationary and $\beta$-mixing
process. Assume that (\ref{Cond_beta_sum1}) holds. Assume that for some
$\vartheta\in\mathbb{R}$, $\mathcal{F\subset}B_{\infty\infty}^{s+1}\left(
\mathbb{R}^{d},\vartheta\right)  $ is nonempty and bounded, $0\in\mathcal{F}$
and $h\in\mathcal{F}$. Let $\zeta_{x}\left(  x\right)  $ be the marginal
density of $x_{t}$. $\zeta_{x}\left(  x\right)  $ is absolutely continuous
with respect to Lebesgue measure, is continuously differentiable with
derivative $\partial\zeta_{x}\left(  x\right)  /\partial x$ vanishing as
$x\rightarrow\pm\infty$ and $\zeta_{x}\left(  x\right)  ^{-1}\partial\zeta
_{x}\left(  x\right)  /\partial x\in\mathcal{F}.$ Assume that one of the
following conditions hold: \newline(i) $\vartheta\leq-1$ and for some
$\gamma>0$ it follows that $\left\Vert \left\langle \chi_{t}\right\rangle
^{\left(  \gamma-\vartheta-1\right)  /2}\right\Vert _{2,\beta}<\infty,$
$\gamma>s$ and $s>1/2;$ \newline(ii) $\vartheta\leq-1$ and for some $\gamma>0$
it follows that $\left\Vert \left\langle \chi_{t}\right\rangle ^{\left(
\gamma-\vartheta-1\right)  /2}\right\Vert _{2,\beta}<\infty,$ $\gamma<s$ and
$\gamma>1/2.$
\end{condition}

Condition \ref{H_C1} directly leads to the following lemma, which is an
immediate consequence of Theorem \ref{FCLT}. Let
\[
v_{t}=\left[
\begin{array}
[c]{c}%
u_{t}\left(  x_{t}-\mu_{x}\right) \\
\frac{\partial g_{h}\left(  x_{t}\right)  }{\partial x}-\theta_{nl}%
+\delta_{nl}\left(  x_{t}\right)  u_{t}%
\end{array}
\right]
\]
and $\Gamma\left(  h\right)  =\sum_{j=-\infty}^{\infty}E\left[  v_{t}%
v_{t-j}^{\prime}\right]  .$

\begin{lemma}
\label{H_Lemma1}Assume that Condition \ref{H_C1} and \ref{H_C2} hold. Let
$v_{n}\left(  h\right)  $ be defined in (\ref{Hausman_EP}). Then,
$v_{n}\left(  h\right)  \rightsquigarrow v\left(  h\right)  $ where $v\left(
h\right)  $ is a Gaussian process with covariance function $\Gamma\left(
h\right)  $ and a.s. uniformly continuous sample paths.
\end{lemma}

The following high level regularity conditions are similar to conditions
imposed in Newey (1994). Since this section is mostly meant to highlight the
usefulness of the functional central limit theory discussed in this paper the
regularity conditions are high level with regard to the semiparametric
estimators. A full development of these estimators is beyond the scope of this paper.

\begin{condition}
\label{H_C2}Let $u_{t}=y_{t}-E\left[  y_{t}|x_{t}\right]  .$ Then, \newline i)
$E\left[  u_{t}^{2}|x_{t}\right]  =\sigma_{t}^{2}\left(  x_{t}\right)  $ and
\[
E\left[  \sigma_{t}^{2}\left(  x_{t}\right)  \zeta_{x}\left(  x\right)
^{-2}\left(  \partial\zeta_{x}\left(  x\right)  /\partial x\right)
^{2}\right]  <\infty.
\]
\newline ii) Let $\hat{g}$ be a series estimator of $g_{h}.$ Then, there
exists a sequence $\kappa=\kappa_{n}$ such that $\kappa_{n}\rightarrow\infty$
as $n\rightarrow\infty$ and $\sqrt{n}\left\Vert \hat{g}-g_{h}\right\Vert
_{2,\beta}^{2}=o_{p}\left(  1\right)  .$ \newline iii) $1/\sqrt{n}\sum
_{t=1}^{n}\partial\left(  \hat{g}\left(  x_{t}\right)  -g_{h}\left(
x_{t}\right)  \right)  /\partial x-\gamma\left(  x_{t}\right)  =o_{p}\left(
1\right)  .$
\end{condition}

The next lemma establishes the limiting process for the empirical moment
function $m_{n}\left(  \theta_{\kappa}\right)  .$

\begin{lemma}
\label{H_Lemma2}Assume that Conditions \ref{H_C1} and \ref{H_C2} hold. Let
$m_{n}\left(  \theta_{0}\right)  $ be defined in (\ref{Hausman_EmpMom}). Then,
for $h\in\mathcal{F}$ fixed,%
\[
\sqrt{n}m_{n}\left(  \theta_{0}\right)  =v_{n}\left(  h\right)  +\tilde
{b}\left(  h\right)  +o_{p}\left(  1\right)
\]
and
\[
\sqrt{n}m_{n}\left(  \theta_{0}\right)  \rightarrow_{d}v\left(  h\right)
+\tilde{b}\left(  h\right)
\]
where $v\left(  h\right)  $ is a Gaussian process with covariance function
$\Gamma\left(  h\right)  $ and a.s. uniformly continuous sample paths. The
bias term $\tilde{b}\left(  h\right)  $ is defined in (\ref{bg}).
\end{lemma}

The following condition is needed to derive an asymptotic limiting
distribution of the estimators for $\theta_{l}$ and $\theta_{nl}.$ The
estimators exist in closed form which greatly simplifies their analysis. For
the representation of the limiting distribution it is useful to partition
$P=\left[  P_{1},P_{2}\right]  $ where $P_{1}=\left[  x_{1},...,x_{n}\right]
^{\prime}$ and $P_{2}=\left[  \left(  p_{2\kappa}\left(  x_{1}\right)
,...,p_{\kappa\kappa}\left(  x_{1}\right)  \right)  ^{\prime},....,\left(
p_{2\kappa}\left(  x_{n}\right)  ,...,p_{\kappa\kappa}\left(  x_{n}\right)
\right)  ^{\prime}\right]  ^{\prime}.$ Then, following Newey (1994, p.1374) an
explicit formula for $\hat{\theta}_{\kappa}$ is given as
\[
\hat{\theta}_{\kappa}=\hat{Q}^{-1}\left[
\begin{array}
[c]{c}%
P_{1}^{\prime}My\\
\hat{\Psi}^{\prime}\left(  P^{\prime}P\right)  ^{-1}P^{\prime}y
\end{array}
\right]
\]
where $\hat{\Psi}=n^{-1}\sum_{t=1}^{n}\partial P^{\kappa}\left(  x_{t}\right)
/\partial x$ and
\[
\hat{Q}=\left[
\begin{array}
[c]{cc}%
n^{-1}\sum_{t=1}^{n}\left(  x_{t}-\bar{x}\right)  ^{2} & 0\\
0 & 1
\end{array}
\right]  .
\]
The following additional high level conditions are imposed.

\begin{condition}
\label{H_C3}i) For $\kappa$ as specified in Condition \ref{H_C2} it follows
that
\[
\hat{Q}=n^{-1}\sum_{t=1}^{n}\left[
\begin{array}
[c]{cc}%
\left(  x_{t}-\bar{x}\right)  ^{2} & 0\\
0 & 1
\end{array}
\right]  \rightarrow_{p}Q=\left[
\begin{array}
[c]{cc}%
\sigma_{x}^{2} & 0\\
0 & 1
\end{array}
\right]
\]
where $\sigma_{x}^{2}=\operatorname*{Var}\left(  x_{t}\right)  $ and $Q$ is a
fixed, positive definite matrix that does not depend on $g.$\newline ii) It
follows that
\[
\sup_{g\in\mathcal{F}}\left\Vert n^{-1/2}\sum_{t=1}^{n}\hat{m}\left(  \chi
_{t},\theta_{0},g\right)  -m\left(  \chi_{t},\theta_{0},g\right)  \right\Vert
=o_{p}\left(  1\right)  .
\]
\newline iii) Assume that $n^{-1}\sum_{t=1}^{n}\left(  x_{t}-\bar{x}\right)
h\left(  x_{t}\right)  =b\left(  h\right)  +o_{p}\left(  1\right)  .$
\end{condition}

The asymptotic limiting distribution of the estimators for $\theta_{l}$ and
$\theta_{nl}$ under the null of $h=0$ and local alternatives is stated in the
next lemma. This distribution then is used to determine critical values for
the Hausman test statistic.

\begin{lemma}
\label{H_Lemma3}Assume that Conditions \ref{H_C1}, \ref{H_C2} and \ref{H_C3}
hold. Then, it follows that for $h$ fixed,
\[
\sqrt{n}\left(  \hat{\theta}_{\kappa}-\theta_{0}\right)  \rightarrow_{d}%
Q^{-1}\left(  v\left(  h\right)  +\tilde{b}\left(  h\right)  \right)
\]
where $\tilde{b}\left(  h_{0}\right)  =0$ and $Q^{-1}v\left(  h\right)  \sim
N\left(  0,Q^{-1}\Gamma\left(  h\right)  Q^{-1}\right)  $. If in addition,
$E\left[  u_{t}|\mathcal{A}^{t-1}\right]  =0$ and $E\left[  u_{t}^{2}%
|x_{t}\right]  =\sigma^{2}$ where $\sigma^{2}$ is constant and $\sigma^{2}>0,$
then it follows that $Q^{-1}v\left(  h\right)  \sim N\left(  0,\sigma
^{2}Q^{-1}\Lambda\left(  h\right)  Q^{-1}\right)  $ where $\Lambda\left(
h\right)  $ is defined as $\Lambda\left(  h\right)  =E\left[  v_{t}%
v_{t}^{\prime}\right]  .$
\end{lemma}

To form the Hausman statistic assume that $\hat{\Gamma}$ is a consistent
estimator of $\Gamma$ and $\hat{Q}$ is consistent for $Q$ by Condition
\ref{H_C3}. Let $e=\left(  1,-1\right)  ^{\prime}.$ A generalized Hausman
statistic to test the null hypothesis of a linear conditional mean then is
given as
\begin{equation}
\hat{H}_{1}=\frac{n\left(  \hat{\theta}_{l}-\hat{\theta}_{nl}\right)  ^{2}%
}{e^{\prime}\hat{Q}^{-1}\hat{\Gamma}\hat{Q}^{-1}e} \label{Hausman-Stat1}%
\end{equation}
If the additional conditions imposed on $u_{t}$ in Lemma \ref{H_Lemma3} hold
then the test statistic can be simplified to
\begin{equation}
\hat{H}_{2}=\frac{n\left(  \hat{\theta}_{l}-\hat{\theta}_{nl}\right)  ^{2}%
}{\left(  e^{\prime}\hat{Q}^{-1}\hat{\Lambda}\hat{Q}^{-1}e\right)  }.
\label{Hausman-Stat2}%
\end{equation}
The limiting distributions of the two Hausman statistics are summarized in the
following Theorem.

\begin{theorem}
\label{Hausman-Theorem}Assume that Conditions \ref{H_C1}, \ref{H_C2} and
\ref{H_C3} hold. Let $\pi\left(  h\right)  =\psi_{1}-\theta_{nl}=-E\left[
\partial h\left(  x_{t}\right)  /\partial x\right]  .$ Then, $\hat{H}_{1}$
defined in (\ref{Hausman-Stat1}) converges (pointwise for $h$ fixed) to a
non-central $\chi^{2}$ process
\[
\hat{H}_{1}\rightarrow_{d}\chi_{1}^{2}\left(  \tilde{\lambda}_{1}\right)
\]
where for fixed $h,$ $\chi_{1}^{2}\left(  \tilde{\lambda}_{1}\right)  $ is a
non-central chi-square distribution with one degree of freedom and
non-centrality parameter $\tilde{\lambda}_{1}$ and
\[
\tilde{\lambda}_{1}=\frac{b\left(  h\right)  +\pi\left(  h\right)  }%
{\sigma_{x}^{2}\sqrt{e^{\prime}Q^{-1}\Gamma\left(  h\right)  Q^{-1}e}}.
\]
If in addition, $E\left[  u_{t}|\mathcal{A}^{t-1}\right]  =0$ and $E\left[
u_{t}^{2}|x_{t}\right]  =\sigma^{2}$ where $\sigma^{2}$ is constant and
$\sigma^{2}>0,$ then it follows that
\[
\hat{H}_{1}\rightarrow_{d}\chi_{1}^{2}\left(  \tilde{\lambda}_{2}\right)
,\text{ }\hat{H}_{2}\rightarrow_{d}\chi_{2}^{2}\left(  \tilde{\lambda}%
_{2}\right)
\]
where the non-centrality parameter $\tilde{\lambda}_{2}$ is given by
\[
\tilde{\lambda}_{2}=\frac{b\left(  h\right)  +\pi\left(  h\right)  }%
{\sigma_{x}^{2}\sqrt{e^{\prime}Q^{-1}\Lambda\left(  h\right)  Q^{-1}e}}.
\]

\end{theorem}

Theorem \ref{Hausman-Theorem} establishes that under the null hypothesis of a
linear conditional mean of $y_{t}$ the limiting distribution of $\hat{H}_{1}$
and, under additional conditions, of $\hat{H}_{2}$ are asymptotically
$\chi_{1}^{2}.$ For a significance level $\alpha,$ let $c_{\alpha}$ be the
critical value of the central $\chi_{1}^{2}$ distribution, i.e. $\alpha
=\Pr\left(  \chi_{1}^{2}>c_{\alpha}\right)  .$ The null hypothesis of a linear
conditional mean then is rejected if $\hat{H}_{1}>c_{\alpha}$ or $\hat{H}%
_{2}>c_{\alpha}.$

The analysis in Theorem \ref{Hausman-Theorem} also shows how the power of the
test against local alternatives depends on the local alternative $h$ and the
marginal distribution of $x_{t}$. The term $b\left(  h\right)  $ captures the
bias in estimating the coefficient $\psi_{1}$ of the linear term in $g\left(
x\right)  $ by linear regression. The term $\pi\left(  h\right)  $ captures
the discrepancy between the two estimators due to the difference between
$\psi_{1}$ and $\theta_{nl}.$ The asymptotic power function of the test is
given by $\Pr\left(  \chi_{1}^{2}\left(  \tilde{\lambda}_{1}\right)
>c_{\alpha}\right)  $ as $h$ ranges over the set of permissible alternatives.

\section{Conclusion}

The paper combines recent results on bracketing numbers for weighted Besov
spaces with a functional central limit theorem for strictly stationary $\beta
$-mixing processes. It is shown that by specializing the bracketing results to
a particular Hilbert space of relevance to the dependent limit theory,
functional central limit theorems for dependent processes indexed by Besov
classes can be obtained directly. These insights lead to some new results in
function spaces with polynomially decaying functions over unbounded domains
and smooth functions over bounded domains.

It is shown how the limit theory can be used to simplify some proofs in the
analysis of semiparametric estimators and tests. An example of a Hausman test
for linearity is considered in detail. More specifically, the central limit
theorem implies a stochastic equicontinuity property that helps shorten
arguments needed to establish the limiting behavior of the test. The central
limit theory also allows to represent the limiting distribution over a class
of local alternatives under general conditions. Finally, a comparison of two
versions of the test when stronger conditions on the model are imposed is provided.

A number of the conditions imposed in Section \ref{Sec_Hausman} are high
level. A detailed analysis of non-parametric estimation in weighted Besov
spaces is beyond the scope of the paper and left for future research.\newpage

\appendix

\section{Proofs\label{Appendix}}

\begin{proof}
[Proof of Theorem \ref{Theorem Bracketing Besov}]The proof follows the
argument in Nickl and P\"{o}tscher (2007, p.184). Let $N\left(  \delta
,\mathcal{F}\text{,}\left\Vert .\right\Vert _{\infty}\right)  $ be the minimal
covering number of $\mathcal{F}$ with respect to $\left\Vert .\right\Vert
_{\infty}$ and $H\left(  \delta,\mathcal{F}\text{,}\left\Vert .\right\Vert
_{\infty}\right)  =\log N\left(  \delta,\mathcal{F}\text{,}\left\Vert
.\right\Vert _{\infty}\right)  $ the metric entropy for $\mathcal{F}$. From
Nickl and P\"{o}tscher (2007, p.184, Eq.3) it follows that for all
$\vartheta\in\mathbb{R}$ and all $\gamma>0$
\begin{equation}
H\left(  \delta,\mathcal{F},\left\Vert \left\langle x\right\rangle ^{\left(
\vartheta-\gamma\right)  /2}\right\Vert _{\infty}\right)  \precsim\left\{
\begin{array}
[c]{cc}%
\delta^{-d/s} & \text{if }\gamma>s-d/p\\
\delta^{-\left(  \gamma/d+1/p\right)  ^{-1}} & \text{if }\gamma<s-d/p
\end{array}
\right.  \label{UniformBound1}%
\end{equation}
Let $B_{i}$ be closed balls in $C\left(  \mathbb{R}^{d},\left\langle
x\right\rangle ^{\left(  \vartheta-\gamma\right)  /2}\right)  =\left\{
f:f\left(  .\right)  \left\langle x\right\rangle ^{\left(  \vartheta
-\gamma\right)  /2}\in C\left(  \mathbb{R}^{d}\right)  \right\}  $ with radius
$\delta$ (relative to the norm $\left\Vert \left(  .\right)  \left\langle
x\right\rangle ^{\left(  \vartheta-\gamma\right)  /2}\right\Vert _{\infty}$)
covering $\mathcal{F}$. Note that the number of such balls is $N\left(
\delta,\mathcal{F}\text{,}\left\Vert \left(  .\right)  \left\langle
x\right\rangle ^{\left(  \vartheta-\gamma\right)  /2}\right\Vert _{\infty
}\right)  .$ Let $f_{i}$ be the center of $B_{i}.$ Then each $B_{i}$ contains
the functions $f$ such that
\[
\sup_{x\in\mathbb{R}^{d}}\left\vert f\left(  x\right)  -f_{i}\left(  x\right)
\right\vert \left\langle x\right\rangle ^{\left(  \vartheta-\gamma\right)
/2}\leq\delta.
\]
The brackets
\[
\left[  f_{i}\left(  x\right)  -\delta\left\langle x\right\rangle ^{\left(
\gamma-\vartheta\right)  /2},f_{i}\left(  x\right)  +\delta\left\langle
x\right\rangle ^{\left(  \gamma-\vartheta\right)  /2}\right]
\]
are contained in $B_{i}$ and cover $\mathcal{F}$. The $\mathcal{L}_{2,\beta
}\left(  P\right)  $ norm of these brackets is
\[
\left\Vert 2\delta\left\langle x\right\rangle ^{\left(  \gamma-\vartheta
\right)  /2}\right\Vert _{2,\beta}.
\]
First consider the case when $\vartheta>0.$ In that case one can choose
$\gamma=\vartheta.$ Then, $\left\Vert 2\delta\left\langle x\right\rangle
^{\left(  \gamma-\vartheta\right)  /2}\right\Vert _{2,\beta}=\left\Vert
2\delta\right\Vert _{2,\beta}.$ Now note that for the constant function
$\delta$%
\[
Q_{\delta}\left(  u\right)  =\inf\left(  t:P\left(  \left\vert \delta
\right\vert >t\right)  \leq u\right)  =\delta
\]
such that
\[
\left\Vert 2\delta\right\Vert _{2,\beta}^{2}=\sum_{m=0}^{\infty}\int%
_{0}^{\beta_{m}}\left(  Q_{2\delta}\left(  u\right)  \right)  ^{2}du=\left(
2\delta\right)  ^{2}\sum_{m=0}^{\infty}\beta_{m}<\infty
\]
by Condition (\ref{Cond_beta_sum1}). One obtains from Nickl and P\"{o}tscher
(2007, p.184, eq. 4) that
\[
H_{[]}\left(  2\delta\sum_{m=0}^{\infty}\beta_{m},\mathcal{F},\left\Vert
{}\right\Vert _{2,\beta}\right)  \leq H\left(  \delta,\mathcal{F},\left\Vert
{}\right\Vert _{\infty}\right)
\]
such that the result follows immediately from (\ref{UniformBound1}).

When $\vartheta\leq0$ the brackets have size
\[
2\delta\left\Vert \left\langle x\right\rangle ^{\left(  \gamma-\vartheta
\right)  /2}\right\Vert _{2,\beta}<\infty
\]
which is bounded by the conditions of the Theorem. It follows again by Nickl
and P\"{o}tscher (2007, p.184, eq. 4) that
\begin{equation}
H_{[]}\left(  2\delta\left\Vert \left\langle x\right\rangle ^{\left(
\gamma-\vartheta\right)  /2}\right\Vert _{2,\beta},\mathcal{F},\left\Vert
{}\right\Vert _{2,\beta}\right)  \leq H\left(  \delta,\mathcal{F},\left\Vert
\left(  .\right)  \left\langle x\right\rangle ^{\left(  \vartheta
-\gamma\right)  /2}\right\Vert _{\infty}\right)  . \label{BracketingBound1}%
\end{equation}
Then, (\ref{UniformBound1}) delivers the stated result.
\end{proof}

\begin{proof}
[Proof of Corollary \ref{Corollary Bracketing Besov}]From the proof of Theorem
\ref{Theorem Bracketing Besov} the $\mathcal{L}_{2,\beta}\left(  P\right)  $
norm of the brackets is, for all $\gamma>0$ and all $\vartheta\in\mathbb{R}$,
\[
\left\Vert 2\delta\left\langle x\right\rangle ^{\left(  \gamma-\vartheta
\right)  /2}\right\Vert _{2,\beta}\leq2\delta M^{\left(  \gamma-\vartheta
\right)  /2}\sum_{m=0}^{\infty}\beta_{m}<\infty.
\]
Therefore, the bound in (\ref{BracketingBound1}) can be applied and the result
again follows by (\ref{UniformBound1}).
\end{proof}

\begin{proof}
[Proof of Theorem \ref{FCLT}]The result follows from Theorem 1 in DMR once all
of their conditions are verified. First show that $\mathcal{F\in L}_{2,\beta
}\left(  P\right)  .$ Let $\mathcal{L}\left(  \beta\right)  $ be the class of
integer valued random variables with distribution function $G_{\beta}\left(
n\right)  =1-\beta_{n} $ for any $n\in\mathbb{N}$ (see DMR, p. 423). For any
$b\in\mathcal{L}\left(  \beta\right)  $ and some real number $K>0$ it follows
that
\begin{align}
E\left[  bf^{2}\left(  \chi_{t}\right)  \right]   &  =E\left[  b\left\langle
\chi_{t}\right\rangle ^{-\vartheta}\left(  f\left(  \chi_{t}\right)
\left\langle \chi_{t}\right\rangle ^{\vartheta/2}\right)  ^{2}\right]
\label{NP_Prop3}\\
&  \leq\left(  \sup_{x\in\mathbb{R}^{d}}\sup_{f\in\mathcal{F}}\left\vert
f\left(  x\right)  \left\langle x\right\rangle ^{\vartheta/2}\right\vert
\right)  ^{2}E\left[  b\left\langle \chi_{t}\right\rangle ^{-\vartheta}\right]
\nonumber\\
&  \leq K^{2}E\left[  b\left\langle \chi_{t}\right\rangle ^{-\vartheta}\right]
\nonumber
\end{align}
where the first inequality is obtained by applying Proposition 3 of Nickl and
P\"{o}tscher (2007) and because $f\left(  x\right)  \left\langle
x\right\rangle ^{\vartheta/2}\in\mathcal{F}$ by assumption. For any
$f\in\mathcal{F}$ it follows from DMR, Eq. (6.2) and
\begin{align}
\left\Vert f\right\Vert _{2,\beta}  &  =\sup_{b\in\mathcal{L}\left(
\beta\right)  }\sqrt{E\left[  bf^{2}\left(  \chi_{t}\right)  \right]
}\label{2_beta_bound1}\\
&  \leq K\sup_{b\in\mathcal{L}\left(  \beta\right)  }\sqrt{E\left[
b\left\langle \chi_{t}\right\rangle ^{-\vartheta}\right]  }\nonumber
\end{align}
where the inequality uses (\ref{NP_Prop3}). If $\vartheta\geq0$ the
inequality
\[
\left\langle \chi_{t}\right\rangle ^{-\vartheta}\leq1
\]
together with $b\geq0$ leads to
\begin{equation}
\left\Vert f\right\Vert _{2,\beta}\leq K\sup_{b\in\mathcal{L}\left(
\beta\right)  }\sqrt{E\left[  b\right]  }=K\left\Vert 1\right\Vert _{2,\beta
}=K\sqrt{\sum_{m=0}^{\infty}\beta_{m}}. \label{2_beta_bound3}%
\end{equation}
When $\vartheta<0$, (\ref{2_beta_bound1}) leads to
\begin{equation}
\left\Vert f\right\Vert _{2,\beta}\leq K\left\Vert \left\langle \chi
_{t}\right\rangle ^{-\vartheta}\right\Vert _{2,\beta}. \label{2_beta_bound2}%
\end{equation}
Since in this case,
\[
\left\langle \chi_{t}\right\rangle ^{-\vartheta}\geq1
\]
and for any $\gamma>0,$
\[
\left\langle \chi_{t}\right\rangle ^{\gamma-\vartheta}\geq\left\langle
\chi_{t}\right\rangle ^{-\vartheta}%
\]
it follows from (\ref{2_beta_bound2}) that
\begin{equation}
\left\Vert f\right\Vert _{2,\beta}\leq K\left\Vert \left\langle \chi
_{t}\right\rangle ^{\gamma-\vartheta}\right\Vert _{2,\beta}<\infty
\label{2_beta_bound4}%
\end{equation}
which is bounded by assumption. Thus, (\ref{2_beta_bound3}) and
(\ref{2_beta_bound4}) show that $f\in\mathcal{F\subset}B_{pq}^{s}\left(
\mathbb{R}^{d},\vartheta\right)  $ with either $\vartheta\geq0$ or
$\vartheta<0$ and some $\gamma>0$ such that $\left\Vert \left\langle
x\right\rangle ^{\left(  \gamma-\vartheta\right)  /2}\right\Vert _{2,\beta
}<\infty$ implies that $\mathcal{F\in L}_{2,\beta}\left(  P\right)  .$

It remains to be show that
\begin{equation}
\int_{0}^{1}\sqrt{H_{[]}\left(  \delta,\mathcal{F},\left\Vert {}\right\Vert
_{2,\beta}\right)  }d\delta<+\infty. \label{Integral bound}%
\end{equation}
For case (i) Theorem \ref{Theorem Bracketing Besov} implies that
$H_{[]}\left(  \delta,\mathcal{F},\left\Vert {}\right\Vert _{2,\beta}\right)
\precsim\delta^{-d/s}$ such that (\ref{Integral bound}) holds for $d/2s<1.$
For case (ii) Theorem \ref{Theorem Bracketing Besov} implies that
$H_{[]}\left(  \delta,\mathcal{F},\left\Vert {}\right\Vert _{2,\beta}\right)
\precsim\delta^{-\left(  \gamma/d+1/p\right)  ^{-1}}$ such that
(\ref{Integral bound}) holds for $1/2\left(  \gamma/d+1/p\right)  ^{-1}<1.$
Cases (iii) and (iv) follow in the same way. This establishes the result.
\end{proof}

\begin{proof}
[Proof of Corollary \ref{Corollary FCLT}]For any $s>d/p$ fix $\vartheta$ such
that $\vartheta>s-d/p.$ By construction $0<\vartheta<\infty$ and thus
$f\left(  .\right)  \left\langle x\right\rangle ^{\vartheta}$ is bounded for
$x\in\mathfrak{X}$ and $f\left(  .\right)  \left\langle x\right\rangle
^{\vartheta}\in B_{pq}^{s}\left(  \mathfrak{X},\vartheta\right)  $. As in
Nickl and P\"{o}tscher (2007, p.186), conclude that $\mathcal{F}\subseteq
B_{pq}^{s}\left(  \mathfrak{X},\vartheta\right)  .$ The results of Theorem
\ref{FCLT} can now be applied. In particular, using the bound in
(\ref{NP_Prop3}) leads to
\begin{equation}
\left\Vert f\right\Vert _{2,\beta}\leq K\sup_{b\in\mathcal{L}\left(
\beta\right)  }\sqrt{E\left[  b\left\langle \chi_{t}\right\rangle
^{-\vartheta}\right]  }\leq KM^{-\vartheta/2}\sqrt{\sum_{m=0}^{\infty}%
\beta_{m}}<\infty.\nonumber
\end{equation}
The result now follows from the fact that \ref{Integral bound} holds by the
results in Corollary \ref{Corollary Bracketing Besov}.
\end{proof}

\begin{proof}
[Proof of Theorem \ref{FCLT_MB}]From DMR Lemma 2, (S.1) and p. 404 it follows
for $\phi\left(  x\right)  =x^{r}$ with $r>1$ that
\begin{equation}
\sum_{m=1}^{\infty}m^{1/\left(  r-1\right)  }\beta_{m}<\infty
\label{D_FCLT_MB_1}%
\end{equation}
and
\begin{equation}
\left\Vert \left\langle \chi_{t}\right\rangle ^{\left(  \gamma-\vartheta
\right)  /2}\right\Vert _{2r,P}<\infty\label{D_FCLT_MB_2}%
\end{equation}
is sufficient for $\left\Vert \left\langle \chi_{t}\right\rangle ^{\left(
\gamma-\vartheta\right)  /2}\right\Vert _{2,\beta}<\infty.$ Note that
(\ref{D_FCLT_MB_2}) holds since $r\left(  \gamma-\vartheta\right)  >1$ and by
Jensen's inequality
\[
\left\Vert \left\langle \chi_{t}\right\rangle ^{\left(  \gamma-\vartheta
\right)  /2}\right\Vert _{2r,P}^{2r}=E\left[  \left\langle \chi_{t}%
\right\rangle ^{r\left(  \gamma-\vartheta\right)  }\right]  \leq1+E\left[
\left\Vert \chi_{t}\right\Vert ^{2r\left(  \gamma-\vartheta\right)  }\right]
<\infty
\]
where the expectation on the RHS is bounded by assumption. The result now
follows from Theorem \ref{FCLT}.
\end{proof}

\begin{proof}
[Proof of Theorem \ref{FCLT2}]The result follows from DMR (eq 2.11) and (eq.
S.1). In particular, the condition
\begin{equation}
\int_{0}^{1}\sqrt{H_{\left[  {}\right]  }\left(  t,,\left\Vert .\right\Vert
_{2p}\right)  dt}<\infty\label{FCLT2_1}%
\end{equation}
needs to hold. From Nickl and P\"{o}tscher (2007) it follows that under the
stated conditions in (i),
\[
H_{\left[  {}\right]  }\left(  t,,\left\Vert .\right\Vert _{2p}\right)
\precsim t^{-d/s}%
\]
such that (\ref{FCLT2_1}) holds as long as $d/\left(  2s\right)  <1$ or
$1/2<s/d.$ Under conditions (ii) one obtains similarly that
\[
H_{\left[  {}\right]  }\left(  t,,\left\Vert .\right\Vert _{2p}\right)
\precsim t^{-\left(  \gamma/d+1/p\right)  ^{-1}}%
\]
such that (\ref{FCLT2_1}) holds as long as $rp/\left(  \gamma p+d\right)  <1$
or $1/2<\left(  \gamma/d+1/p\right)  .$
\end{proof}

\begin{proof}
[Proof of Lemma \ref{H_Lemma1}]Recall that
\[
m\left(  \chi_{t},\theta_{0},\hat{g}_{\kappa}\right)  =\left[
\begin{array}
[c]{c}%
\left(  y_{t}-\mu_{y}-\psi_{1}\left(  x_{t}-\mu_{x}\right)  \right)  \left(
x_{t}-\mu_{x}\right) \\
\frac{\partial P^{\kappa}\left(  x_{t}\right)  }{\partial x}^{\prime}\hat
{\psi}_{\kappa}-\theta_{nl}%
\end{array}
\right]
\]
and that
\[
E\left[  y_{t}\right]  =\psi_{0}+\psi_{1}E\left[  x_{t}\right]  +\mu_{h}.
\]
where $\mu_{h}=E\left[  h\left(  x_{t}\right)  \right]  .$ It follows that
\begin{align}
v_{n}\left(  h\right)   &  =n^{-1/2}\sum_{t=1}^{n}\left(  m\left(  \chi
_{t},\theta_{0},g_{h}\right)  +\gamma\left(  \chi_{t}\right)  -E\left[
m\left(  \chi_{t},\theta_{0},g_{h}\right)  \right]  \right) \nonumber\\
&  =n^{-1/2}\sum_{t=1}^{n}\left[
\begin{array}
[c]{c}%
\left(  u_{t}+n^{-1/2}\left(  h\left(  x_{t}\right)  -\mu_{h}\right)  \right)
\left(  x_{t}-\mu_{x}\right) \\
\partial g_{h}\left(  x_{t}\right)  /\partial x-\theta_{nl}-f_{x}\left(
x_{t}\right)  ^{-1}\partial f_{x}\left(  x_{t}\right)  /\partial xu_{t}%
\end{array}
\right] \nonumber\\
&  -n^{-1/2}\sum_{t=1}^{n}\left[
\begin{array}
[c]{c}%
n^{-1/2}E\left[  \left(  x_{t}-\mu_{x}\right)  \left(  h\left(  x_{t}\right)
-\mu_{h}\right)  \right] \\
0
\end{array}
\right] \nonumber\\
&  =n^{-1/2}\sum_{t=1}^{n}\left[
\begin{array}
[c]{c}%
u_{t}\left(  x_{t}-\mu_{x}\right) \\
\partial g_{h}\left(  x_{t}\right)  /\partial x-\theta_{nl}-\zeta_{x}\left(
x_{t}\right)  ^{-1}\partial\zeta_{x}\left(  x_{t}\right)  /\partial xu_{t}%
\end{array}
\right] \label{Proof_H_Lemma1_D1}\\
&  -n^{-1}\sum_{t=1}^{n}\left[
\begin{array}
[c]{c}%
\left(  x_{t}-\mu_{x}\right)  \left(  h\left(  x_{t}\right)  -\mu_{h}\right)
-E\left[  \left(  x_{t}-\mu_{x}\right)  \left(  h\left(  x_{t}\right)
-\mu_{h}\right)  \right] \\
0
\end{array}
\right] \nonumber
\end{align}
where
\[
n^{-1/2}\sum_{t=1}^{n}\left[
\begin{array}
[c]{c}%
u_{t}\left(  x_{t}-\mu_{x}\right) \\
\partial g_{h}\left(  x_{t}\right)  /\partial x-\theta_{nl}-\left(  \zeta
_{x}\left(  x_{t}\right)  ^{-1}\partial\zeta_{x}\left(  x_{t}\right)  \right)
/\partial xu_{t}%
\end{array}
\right]  \rightsquigarrow v\left(  h\right)
\]
by Theorem \ref{FCLT}. This follows from $\partial g_{h}\left(  x_{t}\right)
/\partial x=\psi_{1}+n^{-1/2}\partial h\left(  x\right)  /\partial x$ and the
fact that
\[
f\left(  y,x\right)  =\partial g_{h}\left(  x\right)  /\partial x-\theta
_{nl}-\left(  \zeta_{x}\left(  x\right)  ^{-1}\partial\zeta_{x}\left(
x\right)  /\partial x\right)  u\in B_{\infty\infty}^{s}\left(  \mathbb{R}%
^{d},\vartheta\right)
\]
if $h\left(  x\right)  \in B_{\infty\infty}^{s+1}\left(  \mathbb{R}%
^{d},\vartheta\right)  $ and $\zeta_{x}\left(  x\right)  ^{-1}\partial
\zeta_{x}\left(  x\right)  /\partial x\in B_{\infty\infty}^{s}\left(
\mathbb{R}^{d},\vartheta\right)  $\textbf{.} It remains to be shown that the
second term in (\ref{Proof_H_Lemma1_D1}) is $o_{p}\left(  1\right)  .$ Since
$\left(  x_{t}-\mu_{x}\right)  \left(  h\left(  x_{t}\right)  -\mu_{h}\right)
\in B_{\infty\infty}^{s+1}\left(  \mathbb{R}^{d},\vartheta-1\right)  $ it
follows by Nickl and P\"{o}tscher (2007, Theorem 1(2)), a strong law of large
numbers for $\beta$-mixing processes and the arguments in the proof of Theorem
2.4.1. in van der Vaart and Wellner (1996, p. 122) that
\[
\sup_{h\in\mathcal{F}}\left\vert n^{-1}\sum_{t=1}^{n}\left(  x_{t}-\mu
_{x}\right)  \left(  h\left(  x_{t}\right)  -\mu_{h}\right)  -E\left[  \left(
x_{t}-\mu_{x}\right)  \left(  h\left(  x_{t}\right)  -\mu_{h}\right)  \right]
\right\vert =o_{p}\left(  1\right)  .
\]

\end{proof}

\begin{proof}
[Proof of Lemma \ref{H_Lemma2}]The proof closely follows arguments in Newey
(1994, Sections 5 and 6), except for the fact that $\left\Vert .\right\Vert
_{2,\beta}$ norms rather than Sobolev norms are the natural norms to use. This
is because stochastic equicontinuity of the empirical process determining the
limiting distribution is directly tied to the $\left\Vert .\right\Vert
_{2,\beta}$ norm. Let $\hat{m}\left(  \chi_{t},\theta,\hat{g}\right)  =\hat
{m}_{t}\left(  \theta\right)  $ and $m\left(  \chi_{t},\theta,g_{h}\right)
=m_{t}\left(  \theta\right)  .$Consider the expansion
\begin{align}
\sqrt{n}m_{n}\left(  \theta_{0}\right)   &  =n^{-1/2}\sum_{t=1}^{n}\hat{m}%
_{t}\left(  \theta_{0}\right)  =n^{-1/2}\sum_{t=1}^{n}\left(  m_{t}\left(
\theta_{0}\right)  +\gamma\left(  \chi_{t}\right)  \right) \nonumber\\
&  +n^{-1/2}\sum_{t=1}^{n}\left(  \hat{m}_{t}\left(  \theta_{0}\right)
-m_{t}\left(  \theta_{0}\right)  -D\left(  \chi_{t},\hat{g}-g_{h}\right)
\right) \label{H_Lem1_D0b}\\
&  +n^{-1/2}\sum_{t=1}^{n}\left(  D\left(  \chi_{t},\hat{g}-g_{h}\right)
-\gamma\left(  \chi_{t}\right)  \right)  . \label{H_Lem1_D0c}%
\end{align}
Let $A_{n,\varepsilon}=1\left\{  \left\Vert n^{-1/2}\sum_{t=1}^{n}\left(
\hat{m}_{t}\left(  \theta_{0}\right)  -m_{t}\left(  \theta_{0}\right)
+\gamma\left(  \chi_{t}\right)  \right)  \right\Vert >\varepsilon\right\}  $
and $B_{n,\varepsilon}=1\left\{  \left\Vert \hat{g}-g_{h}\right\Vert
_{2,\beta}\leq\varepsilon\right\}  .$ Then,
\begin{align*}
\lim_{\varepsilon\downarrow0}\underset{n\rightarrow\infty}{\lim\sup}E\left[
A_{n,\varepsilon}\right]   &  \leq\lim_{\varepsilon\downarrow0}%
\underset{n\rightarrow\infty}{\lim\sup}E\left[  A_{n,\varepsilon/2}\cap
B_{n,\varepsilon/2}\right] \\
&  +\lim_{\varepsilon\downarrow0}\underset{n\rightarrow\infty}{\lim\sup
}P\left(  \left\Vert \hat{g}-g_{h}\right\Vert _{2,\beta}>\varepsilon/2\right)
\end{align*}
where the second term is zero by Condition \ref{H_C2}(ii). Consequently, all
subsequent arguments are restricted to the set $B_{n,\varepsilon}.$ By the
Markov inequality (\ref{H_Lem1_D0b}) and (\ref{H_Lem1_D0c}) are $o_{p}\left(
1\right)  $ if
\begin{align}
&  E\left\Vert n^{-1/2}%
{\textstyle\sum\nolimits_{t=1}^{n}}
\left(  \hat{m}_{t}\left(  \theta_{0}\right)  -m_{t}\left(  \theta_{0}\right)
-D\left(  \chi_{t},\hat{g}-g_{h}\right)  \right)  \right\Vert
\label{H_Lem1_D1}\\
&  \leq E\left\Vert n^{-1/2}%
{\textstyle\sum\nolimits_{t=1}^{n}}
\hat{m}\left(  \chi_{t},\theta_{0},\hat{g}\right)  -m\left(  \chi_{t}%
,\theta_{0},g\right)  \right\Vert \nonumber\\
&  +\sqrt{n}E\left\Vert m\left(  \chi_{t},\theta_{0},\hat{g}\right)  -m\left(
\chi_{t},\theta_{0},g\right)  -D\left(  \chi_{t},\hat{g}-g_{h}\right)
\right\Vert \nonumber
\end{align}
tends to zero and
\begin{align}
&  \left\Vert n^{-1/2}\sum_{t=1}^{n}\left(  D\left(  \chi_{t},\hat{g}%
-g_{h}\right)  -\gamma\left(  \chi_{t}\right)  \right)  \right\Vert
\nonumber\\
&  \leq\left\Vert n^{-1/2}\sum_{t=1}^{n}\left(  D\left(  \chi_{t},\hat
{g}-g_{h}\right)  -\int D\left(  \chi,\hat{g}-g_{h}\right)  dP\right)
\right\Vert \label{H_Lem1_D2}\\
&  +\left\Vert \int D\left(  \chi,\hat{g}-g_{h}\right)  dP-n^{-1/2}\sum
_{t=1}^{n}\gamma\left(  \chi_{t}\right)  \right\Vert \label{H_Lem1_D3}\\
&  =o_{p}\left(  1\right)  .\nonumber
\end{align}
For (\ref{H_Lem1_D1}) note the first term on the RHS of the inequality is
\[
\sum_{t=1}^{n}\hat{m}\left(  \chi_{t},\theta_{0},\hat{g}\right)  -m\left(
\chi_{t},\theta_{0},g\right)  =\left[
\begin{array}
[c]{c}%
n\left(  \bar{y}-\mu_{y}-\theta_{l}\left(  \bar{x}-\mu_{x}\right)  \right)
\left(  \bar{x}-\mu_{x}\right) \\
0
\end{array}
\right]  .
\]
Since
\begin{align*}
E\left\Vert \left(  \bar{y}-\mu_{y}-\theta_{l}\left(  \bar{x}-\mu_{x}\right)
\right)  \left(  \bar{x}-\mu_{x}\right)  \right\Vert  &  \leq\left(
E\left\Vert \bar{y}-\mu_{y}-\theta_{l}\left(  \bar{x}-\mu_{x}\right)
\right\Vert ^{2}E\left\Vert \bar{x}-\mu_{x}\right\Vert ^{2}\right)  ^{1/2}\\
&  =O\left(  n^{-1}\right)
\end{align*}
it follows that the first term is $O\left(  n^{-1/2}\right)  .$ For the second
term in (\ref{H_Lem1_D1}) note that by the same arguments as in Newey (1994,
p. 1361) it follows that $D\left(  \chi,g\right)  =D\left(  \chi\right)
\partial g\left(  x\right)  /\partial x$ where $D\left(  \chi\right)
=\partial m\left(  \chi_{t},\theta,\phi\right)  /\partial\phi|_{\phi=\partial
g\left(  x\right)  /\partial x}=1.$ This leads to
\begin{equation}
D\left(  \chi,g-g_{n}\right)  =\left[
\begin{array}
[c]{c}%
0\\
\frac{\partial}{\partial x}\left(  g-g_{h}\right)
\end{array}
\right]  \label{D_h}%
\end{equation}
and
\[
\left\Vert \left(  m\left(  \chi,\theta,g\right)  -m\left(  \chi,\theta
,g_{h}\right)  -D\left(  \chi,g-g_{h}\right)  \right)  \right\Vert =0
\]
such that the RHS of (\ref{H_Lem1_D1}) is zero and consequently, the term in
(\ref{H_Lem1_D0b}) is $o_{p}\left(  1\right)  $.

For (\ref{H_Lem1_D2}) consider $D\left(  \chi_{t},g\right)  =f\left(  \chi
_{t}\right)  $ where only the second component is relevant. Thus focus on
\begin{equation}
f\left(  \chi_{t}\right)  =\frac{\partial g\left(  x_{t}\right)  }{\partial x}
\label{Dxp}%
\end{equation}
and where $f\left(  \chi_{t}\right)  $ is in a class of functions indexed by
$g\in\mathcal{F}_{g}\mathcal{\in}B_{\infty\infty}^{s+1}\left(  \mathbb{R}%
\text{,}\vartheta_{g}\right)  $. It follows that $f\in\mathcal{F\subset
}B_{\infty\infty}^{s}\left(  \mathbb{R}\text{,}\vartheta_{g}\right)  $ as long
as $g\in\mathcal{F}_{g}$. By Theorem \ref{FCLT} the empirical process%
\[
v_{n}\left(  f\right)  :=n^{-1/2}\sum_{t=1}^{n}\left(  f\left(  \chi
_{t}\right)  -\int f\left(  \chi_{t}\right)  dP\right)
\]
satisfies $v_{n}\left(  f\right)  \rightsquigarrow v\left(  f\right)  $ where
$v\left(  f\right)  $ is a Gaussian process. Note that Theorem \ref{FCLT} is
established by checking all the conditions for DMR, Theorem 1. That Theorem in
turn is established by establishing stochastic equicontinuity of the process
$v_{n}\left(  f\right)  .$ Now, for $f_{h,t}=\partial g_{h}\left(
x_{t}\right)  /\partial x$ and $f_{t}=\partial g\left(  x_{t}\right)
/\partial x$ it follows by from (\ref{D_h}) that
\[
n^{-1/2}\sum_{t=1}^{n}\left(  D\left(  \chi_{t},g-g_{h}\right)  -\int D\left(
\chi,g-g_{h}\right)  dP_{0}\right)  =n^{-1/2}\sum_{t=1}^{n}\left(
f_{t}-f_{h,t}-\int\left(  f_{t}-f_{h,t}\right)  dP\right)
\]
and
\begin{align}
&  \Pr\left(  \left\Vert n^{-1/2}%
{\textstyle\sum\nolimits_{t=1}^{n}}
\left(  D\left(  \chi_{t},\hat{g}-g_{h}\right)  -\int D\left(  \chi,\hat
{g}-g_{h}\right)  dP\right)  \right\Vert >\delta\right) \nonumber\\
&  \leq\Pr\left(  \sup_{\left\Vert \hat{g}-g_{h}\right\Vert _{2,\beta}%
\leq\epsilon}\left\Vert n^{-1/2}\sum_{t=1}^{n}\left(  f_{t}-f_{h,t}%
-\int\left(  f_{t}-f_{h,t}\right)  dP\right)  \right\Vert >\delta/2\right)
\label{H_Lem1_D4}\\
&  +\Pr\left(  \left\Vert \hat{g}-g_{h}\right\Vert _{2,\beta}>\delta/2\right)
\label{H_Lem1_D5}%
\end{align}
where (\ref{H_Lem1_D4}) tends to zero as $\delta\downarrow0$ by the fact that
$v_{n}\left(  f\right)  $ is stochastically equicontinuous and
(\ref{H_Lem1_D5}) tends to zero as $\delta\downarrow0$ by Condition
\ref{H_C2}(ii). Together (\ref{H_Lem1_D4}) and (\ref{H_Lem1_D5}) establishes
that (\ref{H_Lem1_D2}) is $o_{p}\left(  1\right)  $.

To establish that (\ref{H_Lem1_D3}) is $o_{p}\left(  1\right)  $ the
conditions in Newey (1994, Assumption 5.3) are sufficient: there is a function
$\gamma\left(  \chi_{t}\right)  $ such that
\begin{equation}
E\left[  \gamma\left(  \chi_{t}\right)  \right]  =0, \label{NC4}%
\end{equation}%
\begin{equation}
E\left[  \left\Vert \gamma\left(  \chi_{t}\right)  \right\Vert ^{2}\right]
<\infty, \label{NC5}%
\end{equation}
and for all $\left\Vert \hat{g}-g_{h}\right\Vert _{2,\beta}$ small enough,%
\begin{equation}
n^{-1/2}\sum_{t=1}^{n}\left(  \gamma\left(  \chi_{t}\right)  -\int D\left(
\chi_{t},\hat{g}-g_{h}\right)  dP\right)  \rightarrow^{p}0. \label{NC6}%
\end{equation}
Following Newey (1994, p.1362) use \ref{D_h} and integration by parts to
write
\begin{align*}
E\left[  D\left(  \chi,g\right)  \right]   &  =\int\left(  \frac{\partial
}{\partial x}g\left(  x\right)  \right)  \zeta_{x}\left(  x\right)
dx=-\int\left(  \partial\zeta_{x}\left(  x\right)  /\partial x\right)
\zeta_{x}\left(  x\right)  ^{-1}g\left(  x\right)  \zeta_{x}\left(  x\right)
dx\\
&  =-E\left[  \left(  \partial\zeta_{x}\left(  x\right)  /\partial x\right)
\zeta_{x}\left(  x\right)  ^{-1}g\left(  x\right)  \right]  .
\end{align*}
Let $\tau$ index a path (see Newey, 1994, p.1352 for a definition). Let
$g\left(  x_{t},\tau\right)  $ be the projection of $y_{t}$ on $\mathcal{F}$
for a path $\tau$ (see Newey, 1994, p. 1361). For $\delta\left(  x\right)
=\left(  \partial\zeta_{x}\left(  x\right)  /\partial x\right)  \zeta
_{x}\left(  x\right)  ^{-1}$ it follows by the projection theorem that
$E_{\tau}\left[  \delta\left(  x_{t}\right)  g\left(  x_{t},\tau\right)
\right]  =E_{\tau}\left[  \delta\left(  x_{t}\right)  y_{t}\right]  .$ Then,
Newey (1994, Eq. 4.5) implies that
\[
\partial E\left[  D\left(  \chi,g\left(  \tau\right)  \right)  \right]
/\partial\tau=E\left[  \delta\left(  x_{t}\right)  \left(  y_{t}-g\left(
x_{t}\right)  \right)  S\left(  \chi_{t}\right)  \right]
\]
where $S\left(  \chi_{t}\right)  $ is the score of a regular path (see Newey,
1994, Theorem 2.1). By Newey (1994, Theorem 4.1) the correction term
$\gamma\left(  \chi_{t}\right)  $ is given by
\[
\gamma\left(  \chi_{t}\right)  =\delta\left(  x_{t}\right)  u_{t}.
\]
Then $E\left[  \gamma\left(  \chi_{t}\right)  \right]  =0$ follows immediately
from $E\left[  u_{t}|x_{t}\right]  =0$.

For (\ref{NC5}) note that
\[
E\left[  \left\vert \gamma\left(  \chi_{t}\right)  \right\vert |x_{t}\right]
\leq\delta\left(  x_{t}\right)  E\left[  \left\vert u_{t}\right\vert
|x_{t}\right]
\]
Then,
\[
E\left[  \left\vert \gamma\left(  \chi_{t}\right)  \right\vert ^{2}\right]
\leq E\left[  \delta\left(  x_{t}\right)  ^{2}E\left[  u_{t}^{2}|x_{t}\right]
\right]  \leq E\left[  \delta\left(  x_{t}\right)  ^{2}\sigma_{t}^{2}\left(
x_{t}\right)  \right]  <\infty
\]
where $\sigma_{t}^{2}\left(  x_{t}\right)  =E\left[  u_{t}^{2}|x_{t}\right]  $
and $E\left[  \delta\left(  x_{t}\right)  ^{2}\sigma_{t}^{2}\left(
x_{t}\right)  \right]  $ is bounded by Condition (\ref{H_C2})(i)

Finally, (\ref{NC6}) is satisfied by Condition (\ref{H_C2})(iii). This
establishes that (\ref{H_Lem1_D0b}) and (\ref{H_Lem1_D0c}) are $o_{p}\left(
1\right)  $ and therefore that the first claim of the Lemma holds. The second
part of the Lemma follows from Lemma (\ref{H_Lemma1}).
\end{proof}

\begin{proof}
[Proof of Lemma \ref{H_Lemma3}]The estimator $\hat{\theta}_{\kappa}$ solves
\[
m_{n}\left(  \hat{\theta}_{\kappa}\right)  =n^{-1}\sum_{t=1}^{n}\hat{m}\left(
\chi_{t},\hat{\theta}_{\kappa},\hat{g}_{\kappa}\right)  =0
\]
which means that it can be expressed in closed form as as
\[
\left[
\begin{array}
[c]{c}%
\hat{\theta}_{l}\\
\hat{\theta}_{nl}%
\end{array}
\right]  =n^{-1}\hat{Q}^{-1}\left[
\begin{array}
[c]{c}%
P_{1}^{\prime}My\\
\hat{\Psi}^{\prime}\left(  P^{\prime}P\right)  ^{-1}P^{\prime}y
\end{array}
\right]  .
\]
Using the fact that
\[
\left[
\begin{array}
[c]{c}%
\psi_{1}\\
\theta_{nl}%
\end{array}
\right]  =n^{-1}\hat{Q}^{-1}\left[
\begin{array}
[c]{c}%
\psi_{1}P_{1}^{\prime}MP_{1}\\
n\theta_{nl}%
\end{array}
\right]
\]
it follows that
\begin{align}
\sqrt{n}\left[
\begin{array}
[c]{c}%
\hat{\theta}_{l}-\psi_{1}\\
\hat{\theta}_{nl}-\theta_{nl}%
\end{array}
\right]   &  =\hat{Q}^{-1}\frac{1}{\sqrt{n}}\sum_{t=1}^{n}\left[
\begin{array}
[c]{c}%
\left(  x_{t}-\bar{x}\right)  \left(  y_{t}-\psi_{1}\left(  x_{t}-\bar
{x}\right)  \right) \\
\partial P^{\kappa}\left(  x_{t}\right)  ^{\prime}/\partial x\hat{\psi
}_{\kappa}-\theta_{nl}%
\end{array}
\right] \label{Proof_H_Lemma3_D4}\\
&  =\hat{Q}^{-1}\frac{1}{\sqrt{n}}\sum_{t=1}^{n}\hat{m}_{t}\left(  \chi
_{t},\theta_{0},\hat{g}_{\kappa}\right)  \label{Proof_H_Lemma3_D4a}%
\end{align}
By Condition \ref{H_C3}(i) it follows that $\hat{Q}^{-1}-Q^{-1}=o_{p}\left(
1\right)  .$ Then it follows by Condition \ref{H_C3}(ii) and (iii) that
\[
\sqrt{n}\left(  \hat{\theta}_{\kappa}-\theta_{0}\right)  =Q^{-1}\frac{1}%
{\sqrt{n}}\sum_{t=1}^{n}m_{t}\left(  \chi_{t},\theta_{0},\hat{g}_{\kappa
}\right)  +o_{p}\left(  1\right)  .
\]
The result then follows from Lemmas \ref{H_Lemma1} and \ref{H_Lemma2}.
\end{proof}

\begin{proof}
[Proof of \ref{Hausman-Theorem}]It follows directly from Lemma \ref{H_Lemma3}
that for fixed $h,$
\begin{align*}
\tilde{H}_{1}^{1/2}  &  :=\frac{\sqrt{n}\left(  \hat{\theta}_{l}-\hat{\theta
}_{nl}\right)  }{\sqrt{e^{\prime}\hat{Q}^{-1}\hat{\Gamma}\hat{Q}^{-1}e}}%
=\frac{e^{\prime}\left(  \sqrt{n}\left(  \hat{\theta}_{\kappa}-\theta
_{0}\right)  \right)  +\sqrt{n}e^{\prime}\theta_{0}}{\sqrt{e^{\prime}\hat
{Q}^{-1}\hat{\Gamma}\hat{Q}^{-1}e}}\\
&  \rightarrow_{d}\frac{Q^{-1}v\left(  h\right)  }{\sqrt{e^{\prime}%
Q^{-1}\Gamma\left(  h\right)  Q^{-1}e}}+\frac{e^{\prime}Q^{-1}\tilde{b}\left(
h\right)  -E\left[  \partial h\left(  x\right)  /\partial x\right]  }%
{\sqrt{e^{\prime}Q^{-1}\Gamma\left(  h\right)  Q^{-1}e}}%
\end{align*}
where
\[
\frac{e^{\prime}Q^{-1}\tilde{b}\left(  h\right)  }{\sqrt{e^{\prime}%
Q^{-1}\Gamma\left(  h\right)  Q^{-1}e}}=\frac{b\left(  h\right)  }{\sigma
_{x}^{2}\sqrt{e^{\prime}Q^{-1}\Gamma\left(  h\right)  Q^{-1}e}}%
\]
and
\[
\frac{Q^{-1}v\left(  h\right)  }{\sqrt{e^{\prime}Q^{-1}\Gamma\left(  h\right)
Q^{-1}e}}\sim N\left(  0,1\right)  .
\]
The result follows now from the continuous mapping theorem and the fact that
$\tilde{H}_{1}=\left(  \tilde{H}_{1}^{1/2}\right)  ^{2}$. The result for
$\tilde{H}_{2}$ follows in the same way.
\end{proof}

\end{document}